\documentclass[11pt,reqno,twoside]{amsart}

\usepackage{amssymb}
\usepackage{epstopdf}
\usepackage[asymmetric,top=3.5cm,bottom=3.3cm,left=3.1cm,right=3.1cm]{geometry}
\usepackage{booktabs} 
\usepackage{array} 
\usepackage{paralist} 
\usepackage{verbatim} 
\usepackage{graphicx,subfigure,threeparttable}

\usepackage{tabularx}
\usepackage{amsmath,amsfonts,amsthm,mathrsfs,cite}
\usepackage[usenames]{color}
\usepackage{bm}
\usepackage{tabularx}
\usepackage{longtable}
\usepackage{booktabs} 

\newtheorem{theorem}{Theorem}[section]
\newtheorem{lemma}{Lemma}[section]

\theoremstyle{definition}

\theoremstyle{remark}

\newtheorem{remark}{Remark}[section]
\numberwithin{equation}{section}

\newcommand{\R}{\mathbb{R}}

\newcommand{\rmi}{\mathrm{i}}

\usepackage{amsmath}

\allowdisplaybreaks[4]

\title[Quantitative Direct Sampling]{Quantitative Direct Sampling for Initial Acoustic Sources}

\author{Xiaodong Liu}
\address{State Key Laboratory of Mathematical Sciences, Academy of Mathematics and Systems Science, Chinese Academy of Sciences, Beijing 100190, P. R. China.}
\email{xdliu@amt.ac.cn}

\author{Xianchao Wang}
\address{Corresponding author. School of Mathematics, Harbin Institute of Technology, Harbin, P. R. China.}
\email{xcwang90@gmail.com; xcwang@hit.edu.cn}

\date{} 
\begin{document}

\maketitle

\begin{abstract}

This paper addresses the challenge of quantitatively
reconstructing initial acoustic sources from time-dependent wave
measurements. We introduce novel indicator functions defined
through spacetime integrals of acoustic data and carefully designed auxiliary functions. These indicators are foundational for both proving the uniqueness of source reconstruction and developing a
quantitative direct sampling scheme. Our comprehensive numerical
experiments demonstrate the robustness, accuracy, and computational
efficiency of these methods, highlighting their potential for
practical acoustic imaging applications.

\noindent{\bf Keywords:}~~  inverse initial value problem, uniqueness, direct imaging method, acoustic wave.
\end{abstract}

\section{Introduction}
Recovering unknown initial source distributions from time-dependent acoustic wave measurements is a foundational ill-posed inverse problem in imaging science, with critical real-world applications ranging from industrial noise source localization and coal-mine subsidence monitoring to preclinical thermoacoustic imaging \cite{kuchment2008, Liu2015, Kian}. In practical scenarios, transient acoustic sources (e.g., impulsive structural vibrations or thermoacoustic excitations) generate propagating wave fields recorded by boundary sensors; the core challenge is to invert these noisy measurements to reconstruct the spatial distribution of the initial source — a task complicated by the inherent ill-posedness of the inverse problem and the demand for efficient, quantitative reconstruction methods for real-time imaging.

To formalize this problem, we consider a bounded domain $\Omega \subset \mathbb{R}^d$ ($d=2,3$) with boundary $\partial\Omega$, representing the region containing the unknown initial source. 
When an initial source $S(x)$ is activated (corresponding to the initial velocity distribution for the acoustic field), the resulting radiating wave field $p(x,t)$ satisfies the time-domain wave equation:
\begin{equation}\label{eq:main}
	\begin{aligned}
		&\partial_{tt}p(x,t)-\Delta p(x,t)=0, \quad \quad \quad \quad (x,t)\in  \mathbb{R}^d \times \mathbb{R},\\
		& p(x, 0)=0, \quad \partial_t p(x,0)=S(x), \quad \ \ x\in \mathbb{R}^d.
	\end{aligned}
\end{equation}
This model describes acoustic wave propagation in a homogeneous, isotropic medium, featuring zero initial pressure and an initial velocity field determined by the source $S(x)$. The well-posedness of this forward problem has been rigorously established \cite{CourantHilbert, finch2004determining}, and its unique solution satisfies the following far-field asymptotic expansion \cite{Cakoni2019}
\begin{equation*}
	p(x,t+|x|)=\frac{1}{|x|^{\frac{d-1}{2}}}\left\{p^{\infty}(\hat x, t)+\mathcal{O}\left( \frac{1}{|x|^{\frac{d-1}{2}}}\right)  \right\}, \quad |x|\rightarrow \infty,
\end{equation*}
where $p^{\infty}$ denotes the time-domain far-field pattern of the wave field, and $\hat{x} = x/|x|$ is the unit direction vector. We define two measurement modalities as follows:
\begin{equation*}
\mathcal{M}_{\text{near}}:= \{ p(x,t)\, |\, (x,t)\in \Gamma\times \mathbb{R}  \}\quad\mbox{and}\quad
\mathcal{M}_{\text{far}}:= \{p^{\infty}(\hat x,t)\, |\,  (\hat x,t)\in \mathbb{S}^{d-1}\times \mathbb{R}  \},
\end{equation*}
where $\Gamma  \subset  \mathbb{R}^d \setminus \Omega$ is a closed sensor array surface, and $\mathbb{S}^{d-1}:=\{x\in \mathbb{R}^d\, |\,  |x|=1 \}$ denotes the unit sphere (or circle for $d=2$).
The inverse source problem addressed in this work is to reconstruct $S(x)$ from either $\mathcal{M}_{\text{near}}$ or $\mathcal{M}_{\text{far}}$.

Over the past decades, inverse source problems have emerged as a central research topic in applied mathematics and imaging science, with existing reconstruction techniques broadly classified into frequency-domain and time-domain approaches—each demonstrating distinct advantages and inherent limitations:
\begin{enumerate}
\item \textbf{Frequency-domain methods}: These techniques transform transient acoustic measurements into the frequency domain via Fourier analysis, reducing the time-domain inverse problem to a Helmholtz equation-based inverse source problem \cite{ammari2020, Potthast2006}. Classical frequency-domain approaches include recursive linearization methods \cite{Bao2010, Bao2015, Bao2015rec, Cheng2016, Yuan2023}, Fourier expansion methods \cite{Griesmaier2013, wang2017, zhang2015}, and direct sampling methods \cite{direct-sampling-Liu-Hu, Liu2025,extended source-sparse-Liu-Shi,point source-biharmonic-Liu-Shi-Wang,extended source-near-Liu-Wang, Ji2021, spar-2023LiLiu, Li2023}.
\item \textbf{Time-domain methods}: These approaches operate directly on transient wave measurements, leveraging the causal structure of wave propagation to avoid explicit frequency decomposition. Typical time-domain techniques include time-reversal methods \cite{ammari2013time}, total focusing methods \cite{holmes2005post}, and back-projection methods \cite{kuchment2013radon}, which are well-suited for short-pulse, broadband signals and offer intuitive physical interpretations with low computational overhead.
\end{enumerate}

Notably, frequency-domain approaches provide rigorous, tractable theoretical frameworks for analysis, while time-domain methods excel in practical, real-time imaging scenarios \cite{ammari2025}. However, no existing method fully combines the rigorous theoretical guarantees of frequency-domain methods with the practical efficiency of time-domain approaches, creating a critical gap that motivates the present work to develop a unified hybrid reconstruction framework.

Recent efforts have sought to bridge these two paradigms, extending sampling-based methods (linear sampling \cite{Guo2013}, factorization \cite{Cakoni2019}, MUSIC \cite{Cakoni2017}, and direct sampling method \cite{Geng2025,Guo2024}) to the time domain. However, these time-domain sampling methods remain restricted to qualitative shape reconstruction, face theoretical challenges related to interior eigenvalues \cite{Cakoni2021}, and have not been extended to quantitative source reconstruction. This critical gap between the rigorous theory of frequency-domain methods and the practical efficiency of time-domain approaches, combined with the absence of quantitative time-domain source reconstruction tools motivates the present work.

To address this critical gap, we present a novel quantitative time-domain direct sampling method for reconstructing initial acoustic sources from near-field or far-field measurements.
The proposed method relies solely on spacetime integrals of the time-domain measurements with carefully designed auxiliary functions, enabling efficient computation without solving forward problems—making it uniquely suited for real-time imaging applications. 
A rigorous theoretical foundation is provided via constructive uniqueness results for time-domain inverse source problems, based on both near-field and far-field data.
In particular,  to our best knowledge, this work presents the first time-domain framework that achieves quantitative reconstruction capability while maintaining stringent theoretical consistency. 

The remainder of this paper is organized as follows. Sections 2 and 3 present constructive uniqueness results derived from scattered fields and the corresponding far-field data, respectively. Section 4 reports comprehensive numerical experiments that illustrate the robustness, accuracy, and efficiency of the proposed method in both near-field and far-field regimes. Supplementary uniqueness results for alternative initial source configurations are included in the Appendix.

\section{Uniqueness with the near-field data}
In this section, we establish uniqueness results for reconstructing the source $S$ from near-field measurements taken on $\Gamma:=\partial B_R(0)$, where $B_R(0)$ is a ball centered at $0$ with radius $R$ large enough such that $\overline{\Omega}\subset B_R$. 

We begin with some notation. Let  $\mathscr{S}(\mathbb{R})$ be the Schwartz space and let $\mathscr{S}'(\mathbb{R})$ denote  its dual space of tempered distributions.
We define the distributional pairing as
\begin{equation*}
    \langle f, g \rangle := \int_{-\infty}^{\infty} f(t) g(t)\, \mathrm{d}t,
\end{equation*}
where $f\in \mathscr{S}'(\mathbb{R})$ and $g\in \mathscr{S}(\mathbb{R})$.
For $g\in L^1(\R)$, we define its Fourier transform
\begin{equation*}
\mathscr{F}[g](k):=\int_{-\infty}^{\infty} g(t)\, \mathrm{e}^{\mathrm{i} k t}\, \mathrm{d}t,\quad k\in\R
\end{equation*}
and its inverse Fourier transform
\begin{equation*}
\mathscr{F}^{-1}[g](t):=\frac{1}{2\pi}\int_{-\infty}^{\infty} g(k)\, \mathrm{e}^{-\mathrm{i} k t}\, \mathrm{d}k,\quad t\in\R.
\end{equation*}


By Duhamel's principle, the solution to the system \eqref{eq:main} can be rewritten equivalently as the solution to the following nonhomogeneous problem with zero initial conditions:
\begin{equation}\label{eq:Duham}
	\begin{aligned}
		& \partial_{tt}p(x,t)-\Delta p(x,t)=\delta(t) S(x),  \quad (x,t)\in  \mathbb{R}^d \times \mathbb{R},\\
		& p(x, 0)=0, \quad \partial_t p(x,0)=0, \quad \quad \quad x\in \mathbb{R}^d,
	\end{aligned}
\end{equation}
where $\delta(t)$ denotes the Dirac delta distribution.
The initial conditions for the wave equation are understood in the
distributional sense as traces at $t=0^+$ in the space of
tempered distributions $\mathscr{S}'(\mathbb{R})$.

Denote by $G(x, y; t)$ the fundamental solution of the wave equation
\begin{equation}\label{eq:fundamental}
		G(x, y; t):=
		\begin{cases}
		&\displaystyle\frac{H(t-|x-y|)}{2\pi\sqrt{t^2-|x-y|^2}}, \quad \ d=2, \medskip \\
		&\displaystyle \frac{\delta(t-|x-y|)}{4\pi|x-y|}, \quad \qquad d=3,
	\end{cases}
\end{equation}
where $H(t)$ is the Heaviside step function.
Then the solution of \eqref{eq:main} admits the following representation
\begin{equation}\label{eq:radiating-field}
p(x,t)
		=\int_{-\infty}^{\infty}\int_{\mathbb{R}^d} 
		G(x, y; t-\tau) \delta(\tau)   S(y)\, \mathrm{d}y\, \mathrm{d}\tau 
		=\int_{\mathbb{R}^d} G(x,y; t) S(y)\, \mathrm{d}y.
\end{equation}

To facilitate subsequent theoretical analysis, we recall the  source scattering problem in the frequency domain.
By applying the Fourier transform to the time-dependent wave field $p(x,t)$, we obtain its frequency-domain counterpart
\begin{equation*}
u(x,k):=\mathscr{F}[p(x,\cdot)](k), 
\end{equation*}
in terms of which the system \eqref{eq:Duham} reduces to the Helmholtz equation subject to the Sommerfeld radiation condition
\begin{equation*}
	\begin{aligned}
		& -\Delta u(x,k)- k^2 u(x,k)=S(x),  \quad (x,k)\in  \mathbb{R}^d \times \mathbb{R}^{+},\\
		& \lim_{r\to\infty} r^{\frac{d-1}{2}} 
		\left( \frac{\partial u}{\partial r}-{\rm i}ku \right)=0, \quad \quad r=|x|.
	\end{aligned}
\end{equation*}
The explicit solution to this equation is given by
	\begin{equation}\label{eq:solution-freq}
	u(x,k)=\int_{\mathbb{R}^d} \Phi(x,y; k) \, S(y) \, \mathrm{d}y, 
\end{equation}
where 
\begin{equation}\label{Fundamental-Helhmhotz}
\Phi(x,y; k):=\mathscr{F}[G(x,y; \cdot)](k)=
	\begin{cases}
		&\displaystyle\frac{ \rm i}{4}H_0^{(1)}(k|x-y|), \quad \ d=2, \medskip \\
		&\displaystyle \frac{\mathrm{e}^{{\rm i}k|x-y|}}{4\pi|x-y|}, \qquad \qquad d=3.  
	\end{cases}
\end{equation}
Here, $H_0^{(1)}$ denotes the zeroth-order Hankel function of the first kind.

We now introduce the integral $I(y,z)$ as follows
\begin{equation}\label{I}
I(y,z):=\int_{\partial B_R(0)}\Phi(x, y; k)\frac{\partial F_0(x,z;k)}{\partial \nu_x} \,\mathrm{d}s(x), \quad y,z\in B_R(0),
\end{equation}
where
\begin{equation*}
    F_0(x,z;k):=
    \begin{cases}
		\displaystyle J_0(k|x-z|), \quad  d=2, \\
		\displaystyle j_0(k|x-z|), \quad  \, d=3,
	\end{cases}
\end{equation*}
denotes the (spherical) Bessel function of order zero.
Then the following reciprocity relation holds.
\begin{lemma}\label{I(y,z)=I(z,y)}
For the integral $I$ defined in \eqref{I}, the following reciprocity relation holds:
\begin{equation}
I(y,z)=I(z,y),\quad y,z\in B_R(0).
\nonumber
\end{equation}
\end{lemma}
\begin{proof}
The reciprocity relation for the two-dimensional case has already been proved in \cite{point source-biharmonic-Liu-Shi-Wang}. We show that it also holds in three dimensions.

Recall that the spherical harmonics
\begin{equation*}
    Y_q(\hat{x}):=Y_n^{m}(\hat{x}),\quad  q\in \mathbb{Q}:=\{(m,n)
\,|\, m=-n,\cdots,n, \  n=0,1,2,\cdots\},
\end{equation*}
form a complete orthonormal system of $L^2(\mathbb{S}^2)$. Thus, the products $\overline{Y_q(\hat{y})}Y_{\tilde{q}}(\hat{z})$ form an orthonormal basis for $L^2(\mathbb{S}^2)\times L^2(\mathbb{S}^2)$. This, in turn, implies that
\begin{equation}
I(y, z)=\sum_{q,\tilde{q}\in\mathbb{Q}} a^{q}_{\tilde{q}}(|y|,|z|)\overline{Y_q(\hat{y})}Y_{\tilde{q}}(\hat{z}),\quad y,z\in B_R(0),
\nonumber
\end{equation}
where
\begin{equation}
a^{q}_{\tilde{q}}(|y|,|z|)=\int_{\mathbb{S}^2}\int_{\mathbb{S}^2} I(y,z) Y_q(\hat{y}) \overline{ Y_{\tilde{q}}(\hat{z})} \,\mathrm{d}\hat{y} \,\mathrm{d}\hat{z}.
\label{coefficient}
\end{equation}
Inserting \eqref{I} into \eqref{coefficient},
we arrive at the following expression
\begin{align}
a^{q}_{\tilde{q}}(|y|,|z|)
&=\int_{\mathbb{S}^2}\int_{\mathbb{S}^2} \Bigg[\int_{\partial B_R(0)}\Phi(x, y; k)\frac{\partial j_0(k|z-x|)}{\partial \nu_x}ds(x)\Bigg]Y_q(\hat{y})\overline{Y_{\tilde{q}}(\hat{z})} \,\mathrm{d}s(\hat{y}) \,\mathrm{d}s(\hat{z})  \nonumber\\
&=\int_{\partial B_R(0)}\int_{\mathbb{S}^2}\Bigg[\int_{\mathbb{S}^2}\Phi(x, y; k) Y_q(\hat{y})ds(\hat{y})\Bigg]\frac{\partial j_0(k|z-x|)}{\partial \nu_x} \overline{Y_{\tilde{q}}(\hat{z})} \,\mathrm{d}s(\hat{z}) \,\mathrm{d}s(x)  \nonumber\\
&={\rm i}kj_n(k|y|) h_n^{(1)}(kR)\int_{\partial B_R(0)}\Bigg[\frac{\partial}{\partial\nu_x}\int_{\mathbb{S}^2}j_0(k|z-x|)\overline{Y_{\tilde{q}}(\hat{z})}ds(\hat{z})\Bigg] Y_q(\hat{x}) \,\mathrm{d}s(x)  \nonumber\\
&=4{\rm i} \pi k^2 R^2  h_n^{(1)}(kR) j_{\tilde{n}}^{'}(kR)j_n(k|y|)j_{\tilde{n}}(k|z|)\int_{\mathbb{S}^2}Y_q(\hat{x})\overline{Y_{\tilde{q}}(\hat{x})} \,\mathrm{d}s(\hat{x})  \nonumber\\
&=4{\rm i} \pi k^2 R^2  h_n^{(1)}(kR) j_{\tilde{n}}^{'}(kR)j_n(k|y|)j_{\tilde{n}}(k|z|)\delta^q_{\tilde{q}}, \quad y,z\in B_R(0). \label{a-q}
\end{align}
where $\delta^q_{\tilde{q}}$ denotes the Kronecker delta symbol.
In the third and fourth equalities, we have applied the addition theorems \cite[Sec.~2.4, Theorem~2.11]{Colton}
\begin{equation*}
    \Phi(x, y; k)=
    {\rm i}k\sum_{q\in\mathbb{Q}} h_n^{(1)}(k|x|) Y_q(\hat{x})j_n(k|y|)\overline{Y_{q}(\hat{y})}, \quad |x|>|y|
\end{equation*}
and 
\begin{equation*}
    j_0(k|z-x|)=
    4\pi\sum_{q\in\mathbb{Q}} j_n(k|x|) \overline{Y_q(\hat{x})} j_n(k|z|) Y_{q}(\hat{z}), 
\end{equation*}
respectively.
From \eqref{a-q}, we derive that 
\begin{equation*}
a^{q}_{\tilde{q}}(|y|,|z|)=0\,\,\mbox{for}\,\,q,\tilde{q}\in\mathbb{Q},\, q\neq \tilde{q}\quad\mbox{and}\quad
a^{q}_{q}(|y|,|z|)=a^{q}_{q}(|z|,|y|)\,\,\mbox{for}\,\, q\in\mathbb{Q}. 
\end{equation*}
Noting that the coefficients  $a^{q}_{q}$ are independent of $m$  and applying 
the addition theorem \cite[Sec.~2.3, Theorem~2.9]{Colton}
\begin{equation*}
\sum_{m=-n}^n Y_n^{m}(\hat{x}) \overline{Y_n^{m}(\hat{y})}=\sum_{m=-n}^n Y_n^{m}(\hat{y}) \overline{Y_n^{m}(\hat{x})}=\frac{2n+1}{4\pi} P_n(\cos \theta),
\end{equation*}
where $P_n$ denotes the real-valued Legendre polynomials, and $\theta$ is the angle between $\hat x$ and $\hat{y}$, 
we deduce that 
\begin{equation*}
\begin{aligned}
I(y, z)&=\sum_{q\in\mathbb{Q}} a^{q}_{q}(|y|,|z|)\overline{Y_q(\hat{y})}Y_{q}(\hat{z})\\
&=\sum_{n=0}^{\infty}  a^{n}_{n}(|y|,|z|) \sum_{m=-n}^{n} \overline{Y_n^m(\hat{y})} {Y_{n}^m(\hat{z})}\\
&=\sum_{n=0}^{\infty}  a^{n}_{n}(|y|,|z|) \sum_{m=-n}^{n} \overline{Y_{n}^m(\hat{z})} {Y_n^m(\hat{y})} \\
&=\sum_{q\in\mathbb{Q}} a^{q}_{q}(|z|,|y|)\overline{Y_{q}(\hat{z})}Y_q(\hat{y})\\
&=I(z, y),\quad y,z\in B_R(0).
\end{aligned}
\end{equation*}
This completes the proof.
\end{proof}

Next, we present the uniqueness results for reconstructing the initial source from near-field measurements $ \mathcal{M}_{\text{near}}$ in three and two dimensions, respectively. Note that the uniqueness results have been established in \cite{finch2004determining} under the assumption that the support $\Omega$ is strictly convex. We remove this assumption and give a constructive proof.

\begin{theorem}\label{thm:near3D}
The source function $S\in C_c^{\infty}(\Omega)$ in $\R^3$ admits the representation
	\begin{equation}\label{eq:indicator}
		S(y)  =2   \int_{\Gamma}  \langle\partial_{\nu_x} \partial_t  G(x,y; \cdot), \, p(x,\cdot)
		\rangle \, \mathrm{d}s(x),\quad y\in\R^3.
   \end{equation}
\end{theorem}
\begin{proof}
Since $S\in C_c^{\infty}(\Omega)$, it follows from \eqref{eq:radiating-field} that $p(x,\cdot) \in C_c^{\infty}(\mathbb{R}) \subset \mathscr{S}(\mathbb{R})$ 
(see \cite[Sec.~7.2, Theorem~7]{Evans}). 
Note that the Fourier transform $\mathscr{F}:\mathscr{S}(\mathbb{R})\rightarrow\mathscr{S}(\mathbb{R})$ is an isomorphism
\cite[Sec.~10.1, Lemma~10.2]{Eskin}, we then have
\begin{equation*}
    u(x,k) = \mathscr{F}[p(x,\cdot)](k) \in \mathscr{S}(\mathbb{R}).
\end{equation*}
It is clear that $\partial_{\nu_x} \partial_t  G(x,y; \cdot)\in \mathscr{S}'(\R)$.
Applying the Fourier transform of tempered distributions \cite{Eskin, Ma}, one has 
\begin{align}
    \langle \partial_{\nu_x} \partial_t  G(x,y; \cdot) , \, p(x,\cdot)
		\rangle 
        &= \langle \mathscr{F}^{-1}  \mathscr{F}[\partial_{\nu_x} \partial_t  G(x,y; \cdot)], \, p(x,\cdot)
		\rangle  \nonumber\\
        &=  \langle  \mathscr{F}[\partial_{\nu_x} \partial_t  G(x,y; \cdot)], \, \mathscr{F}^{-1} [p(x,\cdot)]
		\rangle. \label{Fourier-distributions}
\end{align}
From \eqref{Fundamental-Helhmhotz}, we have
\begin{equation}\label{eq:Fourier-fund}
	\mathscr{F}\left[\frac{\delta(t-|x-y|)}{4\pi|x-y|}  \right](k)
	=\frac{\mathrm{e}^{\mathrm{i}k|x-y|}}{4\pi|x-y|}
	= \frac{\mathrm{i}k}{4\pi}h_0^{(1)}(k|x-y|),
\end{equation}
where $h_0^{(1)}$ denotes the zeroth-order spherical Hankel function of the first kind. 
Furthermore, with the help of \eqref{Fourier-distributions} and the fact that $\mathscr{F}[\partial_t g](k)=-\mathrm{i}k \mathscr{F}[g](k)$,
we  derive that
\begin{equation*}\label{eq:parseval3D}
	\begin{aligned}
		&2\int_{\Gamma}  \langle\partial_{\nu_x} \partial_t  G(x, y; \cdot), \, p(x,\cdot)
		\rangle \, \mathrm{d}s(x)\\
		&=2\int_{\Gamma} \frac{1}{2\pi}\int_{-\infty}^{\infty}  \, \partial_{\nu_x}\left( -\mathrm{i}k
			 \frac{\mathrm{i}k}{4\pi}h_0^{(1)}(k|x-y|)
		\right) \, \overline{u(x,k)}\,\mathrm{d}k\, \mathrm{d}s(x)\\
		&=\frac{1}{4\pi^2}\int_{-\infty}^{\infty}  \int_{\Gamma} \partial_{\nu_x}
		h_0^{(1)}(k|x-y|)\,  \overline{u(x,k)} k^2 \, \mathrm{d}s(x) \,\mathrm{d}k:={\rm RHS}.
	\end{aligned}
\end{equation*}
In the second equality, we have applied Fubini's theorem based on the facts $u(x,\cdot) \in \mathscr{S}(\mathbb{R})$ and
$\big|\partial_{\nu_x} h_0^{(1)}(k|x-y|) \, k^2 \big| \le C k$ for a positive constant $C$ independent of $k$.
	
Taking the imaginary part on both sides of \eqref{eq:solution-freq}, we obtain
	\begin{equation*}
 \Im (u(x,k)) =\frac{k}{4\pi}\int_{\mathbb{R}^3} j_0(k|x-z|) S(z) \, \mathrm{d}z,\quad x\in\R^3, \,k\in\R^+.
\end{equation*}
Then we have
\begin{align}
&\text{RHS}=\frac{1}{2\pi^2}  \int_{0}^{\infty} \int_{\Gamma} \Re\left[ \partial_{\nu_x}
		h_0^{(1)}(k|x-y|) \,  \overline{u(x,k)} \right]  k^2  \,   \mathrm{d}s(x) \, \mathrm{d}k  \nonumber\\
&=\frac{1}{2\pi^2} \int_{0}^{\infty}   \int_{\Gamma} \left[ \frac{\partial j_0(k|x-y|)}{\partial \nu_x}u(x,k)-\mathrm{ i}\Im (u(x,k) ) \frac{\partial h_0^{(1)}(k|x-y|)}{\partial \nu_x}\right]  \, \mathrm{d}s(x)\,  k^2\,  \mathrm{d}k  \nonumber\\
&=\frac{1}{2\pi^2} \int_{0}^{\infty}   \int_{\Gamma} \int_{\mathbb{R}^3} \left[\frac{\partial j_0(k|x-y|)}{\partial \nu_x}\Phi(x,z;k) - j_0(k|x-z|) \frac{\partial \Phi(x, y; k) }{\partial \nu_x}\right] S(z)\mathrm{d}z \, \mathrm{d}s(x)\,  k^2\,  \mathrm{d}k  \nonumber\\    
&=\frac{1}{2\pi^2} \int_{0}^{\infty}    \int_{\mathbb{R}^3} 
 \int_{\Gamma}\left[\frac{\partial j_0(k|x-z|)}{\partial \nu_x}\Phi(x,y;k) - j_0(k|x-z|) \frac{\partial \Phi(x, y; k) }{\partial \nu_x}\right]  \, \mathrm{d}s(x) S(z) \mathrm{d}z\,  k^2\,  \mathrm{d}k  \nonumber\\
		&=\frac{1}{2\pi^2} \int_{0}^{\infty} \int_{\mathbb{R}^3}j_0(k|y-z|) S(z)\, \mathrm{d}z \, k^2\, \mathrm{d}k.\label{RHS-noepsilon}
\end{align}
In the fourth equality, the interchange of integration is justified by $S \in C_c^{\infty}(\Omega)$ and Lemma \ref{I(y,z)=I(z,y)} is applied. In the fifth equality, we have used the Green representation:
\begin{equation*}
   j_0(k|y-z|) = \int_{\Gamma} \big[\partial_{\nu_x} j_0(k|x-z|) \Phi(x, y; k) - j_0(k|x-z|) \partial_{\nu_x}\Phi(x, y; k)\big] \,    \mathrm{d}s(x), \quad y\in \Omega. 
\end{equation*}

In \cite{radon transform-07}, it is shown that ${\rm RHS}=S$ from \eqref{RHS-noepsilon}. Here, we provide an alternative proof.
We rewrite \eqref{RHS-noepsilon} as
\begin{equation*}
   \text{RHS}= \lim_{\epsilon\to 0^+} \frac{1}{2\pi^2} \int_{0}^{\infty} \int_{\mathbb{R}^3}j_0(k|y-z|) S(z)\,  \, k^2\,  {\rm e}^{-\epsilon k}\mathrm{d}z \mathrm{d}k. 
\end{equation*}
By applying Fubini's theorem and using the integration formula \cite[3.944(11)]{Gradshteyn}
\begin{equation*}
\int_0^{\infty} k {\rm e}^{-\beta k} \sin(b k) \, {\rm d}k=- \frac{\partial}{\partial \beta}\left( \frac{b}{b^2+\beta^2} \right), \quad b>0, \, \beta>0,
\end{equation*}
we obtain
\begin{equation*}
\begin{aligned}
   \text{RHS} &= \lim_{\epsilon\to 0^+} \frac{1}{2\pi^2} \int_{\mathbb{R}^3} \int_{0}^{\infty} j_0(k|y-z|)  \, k^2\,  {\rm e}^{-\epsilon k} \mathrm{d}k S(z)\, \mathrm{d}z\\
   &= \lim_{\epsilon\to 0^+} \frac{1}{2\pi^2} \int_{\mathbb{R}^3} \int_{0}^{\infty} \frac{\sin(k|y-z|)}{|y-z|}  \, k\,  {\rm e}^{-\epsilon k} \mathrm{d}k S(z)\, \mathrm{d}z\\
   &=\lim_{\epsilon\to 0^{+}} \frac{1}{2\pi^2} \int_{\mathbb{R}^3} \frac{2\epsilon}{(|y-z|^2+\epsilon^2)^2} \,  S(z)\, \mathrm{d}z\\
   &=S(y),\quad y\in\R^3,
\end{aligned}
\end{equation*}
where the last equality follows from the  Poisson kernel for the upper half-space in 3D \cite[Sec. 2.1, (15)]{Stein}
\begin{equation*}
    \delta(y-z)=\lim_{\epsilon\to 0^{+}} \frac{\epsilon}{\pi^2(|y-z|^2+\epsilon^2)^2}. 
\end{equation*}
This completes the proof.
\end{proof}

\begin{remark}
Theorem \ref{thm:near3D} establishes the relationship between the source and the boundary measurements via the following distributional integral:
\begin{equation*}
	S(y) = 2  \int_{\Gamma} \int_{-\infty}^{\infty} \partial_{\nu_x} \partial_t  G(x,y; t) \,  p(x,t)\,
		 \mathrm{d}t
		\,\mathrm{d}s(x), \quad y\in \R^3.
\end{equation*} 
We conjecture that this result also holds in the two-dimensional case. However, due to the failure of the Huygens' principle in two dimensions, the wave field exhibits temporal tail effect. Consequently, 
$p(x,\cdot)\notin \mathscr{S}(\R^2)$, and therefore the equality \eqref{Fourier-distributions}
does not hold in the two-dimensional case. 
In the following theorem, we provide an alternative representation using the Fourier transform applied to the scattered fields.
\end{remark}

\begin{theorem}\label{thm:near2D}
The source function $S\in C_c(\Omega)$ in $\R^2$ can be represented by 
	\begin{equation}\label{eq:indicator}
		S(y) = \frac{1}{2\pi} \int_0^{\infty} \int_{\Gamma} \Re\left[\partial_{\nu_x}
   		H_0^{(1)}(k|x-y|)\overline{\mathscr{F}[p(x,\cdot)](k)}\right]\, k\,\mathrm{d}s(x)\,\mathrm{d}k,\quad y\in\R^2.
	\end{equation}
\end{theorem}
\begin{proof}
By the reciprocity relation (Lemma \ref{I(y,z)=I(z,y)}) and the imaginary part of $u(x,k)$ 
\begin{equation*}
	\Im (u(x,k)) =\frac{1}{4}\int_{\mathbb{R}^2} J_0(k|x-z|) S(z) \, \mathrm{d}z,\quad x\in\R^2, \,k\in\R^+,
\end{equation*}
the right-hand side of \eqref{eq:indicator} is given by
\begin{equation*}
	\begin{aligned}
		&\text{RHS} =\frac{1}{2\pi} \int_{0}^{\infty} \int_{\Gamma}  \Re\left[\partial_{\nu_x}
		H_0^{(1)}(k|x-y|) \overline{u(x,k)}\right]\, k\,\mathrm{d}s(x) \,\mathrm{d}k \\
        &=\frac{1}{2\pi}    \int_{0}^{\infty}  \int_{\Gamma} \left[ \frac{\partial J_0(k|x-y|)}{\partial \nu_x}u(x,k)-\mathrm{ i}\Im (u(x,k) ) \frac{\partial H_0^{(1)}(k|x-y|)}{\partial \nu_x}\right]  \mathrm{d}s(x) \, k\,  \mathrm{d}k\\
&=\frac{1}{2\pi}    \int_{0}^{\infty}   \int_{\Gamma}  \int_{\mathbb{R}^2} \left[ \frac{\partial J_0(k|x-y|)}{\partial \nu_x} \Phi(x, z; k) - J_0(k|x-z|) \frac{\Phi(x, y; k)}{\partial \nu_x}  \right] S(z) \,\mathrm{d}z\, \mathrm{d}s(x) \,  k\,  \mathrm{d}k\\
&=\frac{1}{2\pi}    \int_{0}^{\infty}   \int_{\Gamma}  \int_{\mathbb{R}^2} \left[ \frac{\partial J_0(k|x-z|)}{\partial \nu_x} \Phi(x, y; k) - J_0(k|x-z|) \frac{\Phi(x, y; k)}{\partial \nu_x}  \right]  S(z) \,\mathrm{d}z\, \mathrm{d}s(x) \,  k\,  \mathrm{d}k\\
&=\frac{1}{2\pi} \int_{0}^{\infty} \int_{\mathbb{R}^2}J_0(k|z-y|) S(z)\, \mathrm{d}z \, k\, \mathrm{d}k,
\end{aligned}
\end{equation*}
where the last equation follows from the interaction the order of integration as $S \in C_c(\Omega)$ and Green's formula 
\begin{equation*}
   J_0(k|y-z|) = \int_{\Gamma} \partial_{\nu_x} J_0(k|x-z|) \Phi(x, y; k) - J_0(k|x-z|) \partial_{\nu_x}\Phi(x, y; k) \,    \mathrm{d}s(x), \quad y\in \Omega. 
\end{equation*} 

According to the Hankel transform \cite[ (10.22.49) ]{Olver}, one has 
\begin{equation*}
\int_{0}^{\infty} k  {\rm e}^{-\epsilon k}   J_0(k|y-z|) \, \mathrm{d}k=\frac{\epsilon}{(\epsilon^2+|y-z|^2)^{3/2}}.
\end{equation*}
By similar arguments to the three dimensional case, one can derive that 
\begin{equation*}
\begin{aligned}
       \text{RHS}&= \lim_{\epsilon\to 0^+} \frac{1}{2\pi} \int_{0}^{\infty} \int_{\mathbb{R}^2} J_0(k|y-z|) S(z)\,  \, k\,  {\rm e}^{-\epsilon k}\, \mathrm{d}z \,\mathrm{d}k\\
       &= \lim_{\epsilon\to 0^+} \frac{1}{2\pi} \int_{\mathbb{R}^2} \frac{\epsilon}{(\epsilon^2+|y-z|^2)^{3/2}}  S(z) \mathrm{d}z\\
       &=S(y), \quad y\in\R^2,
\end{aligned}
\end{equation*}
where the last equality follows from the  Poisson kernel for the upper half-space in 2D \cite[Sec. 2.1, (15)]{Stein}
\begin{equation*}
    \delta(y-z)=\lim_{\epsilon\to 0^{+}} \frac{\epsilon}{2\pi(|y-z|^2+\epsilon^2)^{3/2}}. 
\end{equation*}
This completes the proof.
\end{proof}

In contrast to Theorem \ref{thm:near3D}, the analogous result of Theorem \ref{thm:near2D} for $\R^3$ is formulated in the following theorem.
\begin{theorem}\label{thm:near3D-2}
The source function $S\in C_c(\Omega)$ in $\R^3$ has a representation
	\begin{equation*}\label{eq:indicator3}
		S(y) = \frac{1}{2\pi^2} \int_{0}^{\infty}\int_{\Gamma} \Re\left[  \partial_{\nu_x}
		h_0^{(1)}(k|x-y|) \,  \overline{\mathscr{F}[p(x,\cdot)](k)} \right]\, k^2  \, \mathrm{d}s(x)\mathrm{d}k\, ,\quad y\in \R^3.
	\end{equation*}
\end{theorem}


\section{Uniqueness with the far-field data}
In this section, we establish uniqueness results for the reconstruction of the source from far-field measurements. To this end, we first present a representation of the time-domain far-field pattern. 

\begin{theorem}\label{thm:far-field} 
Let $S\in C_c(\Omega)$.
Then the time-domain far-field pattern $ p^{\infty}$ admits the representation
\begin{equation}\label{eq:far_forward}
  p^{\infty}(\hat x, t)=
  	\begin{cases}
  &\displaystyle  \frac{1}{2\sqrt{2}\pi}\int_{\mathbb{R}^2} \frac{H(t+\hat x \cdot y) }{\sqrt{t+\hat x \cdot y}}S(y)\,  \mathrm{d}y, \quad \, d=2,\bigskip\\
  &\displaystyle\frac{1}{4\pi}\int_{\mathbb{R}^3} \delta(t+\hat x \cdot y) S(y)\, \mathrm{d}y,  \qquad \quad \, d=3.
  \end{cases}
\end{equation}
\end{theorem}
\begin{proof}
	
We consider only the two-dimensional case, as the three-dimensional case can be treated by a similar argument.

From
\begin{equation*}
	\begin{aligned}
  |x-y|&
  =|x|\sqrt{1-2\frac{\hat{x}\cdot y}{|x|}+\left(\frac{|y|}{|x|}\right)^2}
  =|x|-\hat{x}\cdot y+\mathcal{O}\left( \frac{1}{|x|}\right), \quad |x|\to \infty
  \end{aligned}
\end{equation*}
we derive 
\begin{equation*}
	\sqrt{(t+|x|)^2-|x-y|^2}=\sqrt{2t|x|+2x\cdot y+t^2-|y|^2}=\sqrt{2|x|}\sqrt{t+\hat{x}\cdot y+ \mathcal{O}(1/|x|)}.
\end{equation*}
Inserting the above expansion into \eqref{eq:radiating-field}, the theorem follows. 
\end{proof}

\begin{remark}
The measurement time for the near field 
$p(x,t)$ starts naturally from zero, corresponding to the activation of the source. In contrast, for the far-field pattern $p^{\infty}(\hat{x},t)$, the effective time reference is shifted by the propagation delay, so it starts from $\min_{y\in \Omega}\hat{x}\cdot y$ for any fixed observation direction $\hat x$. 
\end{remark}

\begin{remark}[Translation Property]
It is noted that the time-domain far-field pattern satisfies the translation property: if $\tilde{S}(\cdot) = S(\cdot - z_0)$ for some fixed vector $z_0 \in \mathbb{R}^d$, then 
\begin{equation*}
    p_{\tilde{S}}^{\infty}(\hat{x}, t) = p_{S}^{\infty}(\hat{x}, t + \hat{x} \cdot z_0).
\end{equation*}
\end{remark}

Next, we show the one-to-one correspondence between the time-domain and frequency-domain far-field patterns via the Fourier transform. To this end, we first recall the frequency-domain far field pattern of the Helmholtz equation. 
From the asymptotic behavior of the Hankel function, the solution to the Helmholtz equation satisfies
\begin{equation*}
u(x, k) = \frac{\mathrm{e}^{{\rm i}k|x|}}{|x|^{\frac{d-1}{2}}} \left\{ u_\infty(\hat{x}, k) + \mathcal{O}\left( \frac{1}{|x|} \right) \right\}, \quad |x| \to \infty,
\end{equation*}
uniformly for all directions $\hat{x} = x/|x|$, where the frequency-domain far-field pattern $u_\infty$  is defined by
\begin{equation}\label{eq:far-freq1}
		u_\infty(\hat{x}, k)=\frac{{\rm e}^{{\rm i}\frac{\pi}{4}}}{\sqrt{8k\pi}} \left( {\rm e}^{-{\rm i}\frac{\pi}{4}} \sqrt{\frac{k}{2\pi}} \right)^{d-2}\int_{\mathbb{R}^d} \mathrm{e}^{-\mathrm{i} k \hat x\cdot y} S(y)\, \mathrm{d}y,  \quad d=2, 3.
\end{equation}

\begin{theorem} 
Let $p^{\infty}(\hat{x},t)$ and $u^{\infty}(\hat{x}, k)$ denote the far-field patterns defined in \eqref{eq:far_forward} and \eqref{eq:far-freq1}, respectively. 
Then 
\begin{equation*}
u^{\infty}(\hat x, k)=\mathscr{F}[p^{\infty}(\hat x,t)](k),\quad k\in \mathbb{R}^+.
\end{equation*}
\end{theorem}
\begin{proof}
Using the time-shift property of the Fourier transform 
\begin{equation*}
	\mathscr{F}[g(t + \tau)](k) = \mathrm{e}^{-\rmi k\tau}\, \mathscr{F}[g(t)](k),
\end{equation*}
together with the formula
\begin{equation*}
	\int_0^\infty \frac{{\rm e}^{{\rm i}kt}}{\sqrt{t}}\,\, {\rm d}t  = \Gamma\left(\frac{1}{2}\right) {\rm e}^{\rmi \frac{\pi}{4}}  k^{-1/2}=\frac{\sqrt{\pi}\mathrm{e}^{\mathrm{i}\pi/4}}{\sqrt{k}}, \quad k\in \mathbb{R}^+,
\end{equation*}
we derive in the two-dimensional case that
\begin{equation*}
	\begin{aligned}
		\mathscr{F}[p^{\infty}(\hat x,t)](k)
		&=\mathscr{F}\left[\frac{1}{2\sqrt{2}\pi}\int_{\mathbb{R}^2} \frac{H(t+\hat x \cdot y) }{\sqrt{t+\hat x \cdot y}}S(y)\, \mathrm{d}y \right](k)\\
		&=\frac{1}{2\sqrt{2}\pi}\int_{\mathbb{R}^2}   \mathrm{e}^{-\mathrm{i} k \hat x\cdot y}  \mathscr{F}\left[ \frac{H(t) }{\sqrt{t}} \right](k) \, S(y) \,\mathrm{d}y\\
		&=\frac{1}{2\sqrt{2}\pi}\int_{\mathbb{R}^2}   \mathrm{e}^{-\mathrm{i} k \hat x\cdot y}  	\int_0^\infty \frac{{\rm e}^{{\rm i}kt}}{\sqrt{t}}\,\, {\rm d}t \, S(y) \,\mathrm{d}y\\
		&=\frac{\mathrm{e}^{\mathrm{i}\pi/4}}{\sqrt{8\pi k}}\int_{\mathbb{R}^2} \mathrm{e}^{-\mathrm{i} k \hat x\cdot y} S(y)\, \mathrm{d}y =u^{\infty}(\hat x, k).
	\end{aligned}
\end{equation*}

For the three-dimensional case, using the fact that $\mathscr{F} [\delta(t)](k)=1$, we obtain 
\begin{equation*}
  \begin{aligned}
  \mathscr{F}[p^{\infty}(\hat x,t)](k)
  &=\mathscr{F}\left[\frac{1}{4\pi}\int_{\mathbb{R}^3} \delta(t+\hat x \cdot y) S(y)\, \mathrm{d}y\right](k)\\
   &=\frac{1}{4\pi}\int_{\mathbb{R}^3}  \mathrm{e}^{-\mathrm{i} k \hat x\cdot y}   \mathscr{F}\left[ \delta(t) \right](k)\, S(y)\, \mathrm{d}y\\
   &=\frac{1}{4\pi}\int_{\mathbb{R}^3}  \mathrm{e}^{-\mathrm{i} k \hat x\cdot y}   S(y) \mathrm{d}y =u^{\infty}(\hat x, k),
  \end{aligned}
\end{equation*}
which completes the proof. 
\end{proof}

We now present the uniqueness results for recovering the source function $S$ from the far-field pattern $\mathcal{M}_{\text{far}}$.

\begin{theorem}\label{thm:far23D} 
	For $S\in C_c(\Omega)$,  the following identity holds:
\begin{equation}\label{eq:indicator-far23D}
    S(y) =
     \frac{1}{2\pi^2} \left(\frac{2}{\mathrm{i}\pi}\right)^{\frac{3-d}{2}} \int_{\mathbb{R}^d}  \int_{-\infty}^{\infty} p^{\infty}\left(\hat x, t-\hat x\cdot y \right) |x|^{\frac{3-d}{2}}\mathrm{e}^{\mathrm{i}|x|t } \, \mathrm{d}t\,\mathrm{d}x, \quad y\in \R^d.
\end{equation} 
\end{theorem}
\begin{proof}
We first consider the three-dimensional case. Since 
	\begin{equation*}
		p^{\infty}(\hat x, t)=\frac{1}{4\pi}\int_{\mathbb{R}^3} \delta(t+\hat x \cdot z) S(z)\, \mathrm{d}z,
	\end{equation*}
substituting it into the right-hand side of   \eqref{eq:indicator-far23D}, we have
\begin{equation*}
		\text{RHS}=\frac{1}{2\pi^2}  \int_{\mathbb{R}^3} \int_{-\infty}^{\infty}  p^{\infty}\left(\hat x, t-\hat x\cdot y \right) \mathrm{e}^{\mathrm{i}|x|t } \, \mathrm{d}t\, \mathrm{d}x,\quad y\in\R^3, 
	\end{equation*}
which implies that
	\begin{equation*}
		\begin{aligned}
			 \text{RHS}&=\frac{1}{2\pi^2}  \int_{\mathbb{R}^3} \int_{-\infty}^{\infty}  \left(\frac{1}{4\pi}\int_{\mathbb{R}^3} \delta(t+ \hat x \cdot z - \hat x \cdot y) S(z)\, \mathrm{d}z\right) \mathrm{e}^{\mathrm{i}|x|t } \, \mathrm{d}t\, \mathrm{d}x\\
			&=\frac{1}{(2\pi)^3}  \int_{\mathbb{R}^3} \int_{\mathbb{R}^3}  \mathrm{e}^{-\mathrm{i}x\cdot ( z- y) } S(z)\, \mathrm{d}z\, \mathrm{d}x\\
            &=S(y),\quad y\in\R^3.
		\end{aligned}
	\end{equation*}


Next, we consider the two-dimensional case.
	Substituting the far-field pattern \eqref{eq:far_forward} into the right-hand side of \eqref{eq:indicator-far23D}:
\begin{align*}
	\text{RHS}&:=\frac{\mathrm{e}^{-\mathrm{i}\pi/4}}{\sqrt{2\pi^3}}  \int_{\mathbb{R}^2}\int_{-\infty}^{\infty}   \sqrt{|x|}  \, p^{\infty}\left(\hat x, t-\hat x\cdot y \right) \mathrm{e}^{\mathrm{i}|x|t } \, \mathrm{d}t\, \mathrm{d}x\\
    &=\frac{\mathrm{e}^{-\mathrm{i}\pi/4}}{\sqrt{2\pi^3}}    \int_{\mathbb{R}^2} \sqrt{|x|} \lim_{\eta \to 0^+} \int_{-\infty}^{\infty}  \left(\frac{1}{2\sqrt{2}\pi}\int_{\mathbb{R}^2} \frac{H[t+\hat x \cdot (z - y)] }{\sqrt{t+\hat x \cdot (z -y)}}S(z)\, \mathrm{d}z \right) \mathrm{e}^{\mathrm{i}|x|t } {\rm e}^{-\eta t}\, \mathrm{d}t\, \mathrm{d}x.
\end{align*}

Furthermore, by Fubini’s theorem and the change of variables $\tau:=t+\hat x \cdot (z- y)$, the above expression can be rewritten as
\begin{equation*}
 \text{RHS}=  \frac{\mathrm{e}^{-\mathrm{i}\pi/4}}{4\pi^{\frac{5}{2}}}  \int_{\mathbb{R}^2} 
 \sqrt{|x|}  \lim_{\eta \to 0^+} \int_{\mathbb{R}^2}  
 \mathrm{e}^{-\mathrm{i} x \cdot (z - y)}  \mathrm{e}^{\eta \hat{x} \cdot (z - y)}  \left(\int_{-\infty}^{\infty}  \frac{H(\tau) }{\sqrt{\tau}} \mathrm{e}^{\mathrm{i}|x|\tau}  \mathrm{e}^{-\eta \tau} \mathrm{d}\tau \right)   S(z)\, \mathrm{d}z \, \mathrm{d}x.   
\end{equation*}
Since $H(\tau)$ vanishes for $\tau < 0$, we use the integral representation of the Gamma function \cite[Sec. 6.1, (6.1.1)]{Abramowitz1965}
\begin{equation*}
    \Gamma(z)=\alpha ^{z}\int_0^{\infty} \frac{{\rm e}^{-\alpha t}}{t^{1-z}}  {\rm d}t, \quad  \Re z>0, \ \Re \alpha >0, 
\end{equation*}
and the fact that $\Gamma(1/2)=\sqrt{\pi}$ to derive that 
\begin{equation*}
\begin{aligned}
 \text{RHS}
&= \frac{\mathrm{e}^{-\mathrm{i}\pi/4}}{4\pi^{\frac{5}{2}}}  \int_{\mathbb{R}^2} 
 \sqrt{|x|} \int_{\mathbb{R}^2}  
 \mathrm{e}^{-\mathrm{i} x \cdot (z - y)}  \mathrm{e}^{\eta \hat{x} \cdot (z - y)}   \lim_{\eta \to 0^+} \left(\int_0^{\infty} \frac{\mathrm{e}^{-(\eta-\mathrm{i}|x|) \tau}}{\sqrt{\tau}} \, \mathrm{d}\tau \right)   S(z)\, \mathrm{d}z \, \mathrm{d}x\\
 &= \frac{\mathrm{e}^{-\mathrm{i}\pi/4}}{4\pi^{\frac{5}{2}}}  \int_{\mathbb{R}^2} 
 \sqrt{|x|} \int_{\mathbb{R}^2}  
 \mathrm{e}^{-\mathrm{i} x \cdot (z - y)}  \mathrm{e}^{\eta \hat{x} \cdot (z - y)}   \lim_{\eta \to 0^+} \left(\Gamma(1/2)(\eta-{\rm i} |x|)^{-1/2} \right)   S(z)\, \mathrm{d}z \, \mathrm{d}x\\
 &=  \frac{\mathrm{e}^{-\mathrm{i}\pi/4}}{4\pi^{\frac{5}{2}}}  \int_{\mathbb{R}^2} 
 \sqrt{|x|}  \int_{\mathbb{R}^2}  
 \mathrm{e}^{-\mathrm{i} x \cdot (z - y)}  \sqrt{\pi}|x|^{-1/2} {\rm e}^{{\rm i}\pi/4} S(z)\, \mathrm{d}z \, \mathrm{d}x\\
 &=\frac{1}{(2\pi)^2}  \int_{\mathbb{R}^2} \int_{\mathbb{R}^2}  \mathrm{e}^{-\mathrm{i}x\cdot ( z- y) } S(z)\, \mathrm{d}z\, \mathrm{d}x\\
 &=S(y),\quad y\in\R^2.
 \end{aligned}
\end{equation*}
This completes the proof.
\end{proof}

\section{Time-domain direct sampling method}
Theorem \ref{thm:near3D} suggests that the indicator
\begin{equation*}
	\mathcal{I}_{\text{near}}(y) 
     =2  \int_{\Gamma} 
  \langle \partial_{\nu_x}\partial_t G(x,y; \cdot),  p(x,\cdot) \rangle
		\,\mathrm{d}s(x), \quad y\in \R^3 
\end{equation*}
can be used to quantitatively reconstruct the source function. However, due to the singular nature of the derivatives of the fundamental solution $G$, this indicator must be modified.
A straightforward calculation shows that
\begin{equation*}
	\partial_{\nu_x} \partial_t G(x-y,t) =- \frac{(x-y)\cdot \nu_x}{4\pi |x-y|^2} \left( \delta^{''}(t - |x-y|) + \frac{\delta'(t-|x-y|)}{|x-y|} \right).
\end{equation*}
Accordingly,  the above indicator can be rewritten as
\begin{equation}\label{eq:near3D_time}
	\mathcal{I}^{(1)}_{\mathrm{near}}(y) = \int_{\Gamma} \frac{(x-y)\cdot \nu_x}{2\pi |x-y|^2} 
	\left( \frac{\partial_t p(x,|x-y|)}{|x-y|} - \partial_{tt} p(x,|x-y|)  \right) \mathrm{d}s(x), \quad y\in \R^3.
\end{equation}
For numerical implementation, given noisy measurements $p\in \mathcal{M}_{\text{near}}$, 
we may apply frequency filtering to truncate high-frequency components and smooth the signal, from which the first and second time derivatives $\partial_t p$ and $\partial_{tt} p$ can then be evaluated. As an alternative, we may first compute the Fourier transform of the noisy data $p\in\mathcal{M}_{\text{near}}$ and then, motivated by the explicit representations derived in Theorems \ref{thm:near2D}–\ref{thm:near3D-2}, employ the following indicator:
\begin{equation}\label{I-near}
	\mathcal{I}^{(2)}_{\text{near}}(y) 
    =\begin{cases}
  &\!\!\!\!\!\!\displaystyle \frac{1}{2\pi^2} \int_{0}^{\infty}\int_{\Gamma} \Re\left[  \partial_{\nu_x}
		h_0^{(1)}(k|x-y|) \,  \overline{\mathscr{F}[p(x,\cdot)](k)} \right]\, k^2  \, \mathrm{d}s(x)\mathrm{d}k, \quad y\in \R^3, \medskip\\
  &\!\!\!\!\displaystyle\frac{1}{2\pi} \int_0^{\infty} \int_{\Gamma} \Re\left[\partial_{\nu_x}
   		H_0^{(1)}(k|x-y|)\overline{\mathscr{F}[p(x,\cdot)](k)}\right]\, k\,\mathrm{d}s(x)\,\mathrm{d}k,\quad y\in \R^2.
  \end{cases}
\end{equation} 

Finally, inspired by Theorem \ref{thm:far23D}, we introduce the following indicator based on the time-domain far-field patterns $p^\infty\in\mathcal{M}_{\text{far}}$:
\begin{equation}\label{I-far}
 	\mathcal{I}_{\text{far}}(y)=
     \frac{1}{2\pi^2} \left(\frac{2}{\mathrm{i}\pi}\right)^{\frac{3-d}{2}} \int_{\mathbb{R}^d}  \int_{-\infty}^{\infty} p^{\infty}\left(\hat x, t-\hat x\cdot y \right) |x|^{\frac{3-d}{2}}\mathrm{e}^{\mathrm{i}|x|t } \, \mathrm{d}t\,\mathrm{d}x, \quad y\in \R^d.
 \end{equation}
Note that the evaluation of $\mathcal{I}_{\text{far}}$ involves a volume integral. In subsequent numerical experiments, the observation directions are sampled on an equidistantly spaced square or cubic grid.

We consider four groups of numerical tests, referred to as  {\bf Near2D, Near3D, Far2D}, and {\bf Far3D}. 

\subsection{Synthetic data generation}
To generate synthetic near field data, we employ the finite element method to solve the forward problem of \eqref{eq:main}. 
 
For synthetic far-field data, we make use of the representation \eqref{eq:far_forward} given in Theorem \ref{thm:far-field}.
In $\R^3$, from formula \eqref{eq:far_forward}, the far-field pattern can be rewritten as
\begin{equation*}
	p^{\infty}(\hat{x}, t) = \frac{1}{4\pi} \int_{\hat{x} \cdot y = -t} S(y)\,  \mathrm{d}\sigma(y), \quad \hat{x}\in \mathbb{S}^{2}.
\end{equation*}
In $\R^2$, we have
\begin{equation*}
	\begin{aligned}
		p^{\infty}(\hat{x}, t)
		= \frac{1}{2\sqrt{2}\pi}
		\int_{\hat{x} \cdot y > -t}
		\frac{1}{\sqrt{t + \hat{x} \cdot y}}
		S(y)\,\mathrm{d}y, \quad \hat{x}\in \mathbb{S}^{1}. 
	\end{aligned}
\end{equation*}
To remove the singularity at $t + \hat{x} \cdot y = 0$, we introduce the change of variables
\begin{equation*}
	s := \hat{x} \cdot y, 
	\qquad 
	r := \hat{x}^{\perp} \cdot y,
\end{equation*}
where $\hat{x}^{\perp}$ is the unit vector perpendicular to $\hat{x}$. 
This gives the representation
\begin{equation*}
 y = s\hat{x} + r\hat{x}^{\perp},
\end{equation*}
and the far-field pattern can  be rewritten as
\begin{align*}
p^{\infty}(\hat{x}, t)
&= \frac{1}{2\sqrt{2}\pi}
	\int_{s>-t} \int_{-\infty}^{\infty}
	\frac{1}{\sqrt{t + s}}
	S(s,r)\,\mathrm{d}r\,\mathrm{d}s\\
&= \frac{1}{\sqrt{2}\pi}
	\int_{0}^{\infty} \int_{-\infty}^{\infty}
	S(v^2 - t, r)\,\mathrm{d}r\,\mathrm{d}v, \quad \hat{x}\in \mathbb{S}^{1}.
\end{align*}


To test the robustness of the proposed method, we add signal-to-noise ratio (SNR) noise to the time-dependent measurements, i.e.,  
\begin{equation*}
	m^{\delta}=m+\sqrt{\frac{P_{\text{signal}}(m)}{10^{\mathrm{SNR}/10}} } \times Z,\quad m=p(x,\cdot),\, p^\infty(\hat{x},\cdot),
\end{equation*}
where $P_{signal}(m)=\mathbb{E}(|m|^2)$ stands for the power of the signal, with $\mathbb{E}$ the expectation operator over time $t$, and
 $Z\sim \mathcal{N}(0,1)$ is a standard Gaussian random variable. 

\begin{figure}
	\centering
	\subfigure[signal without noise]{\includegraphics[width=0.45\textwidth]{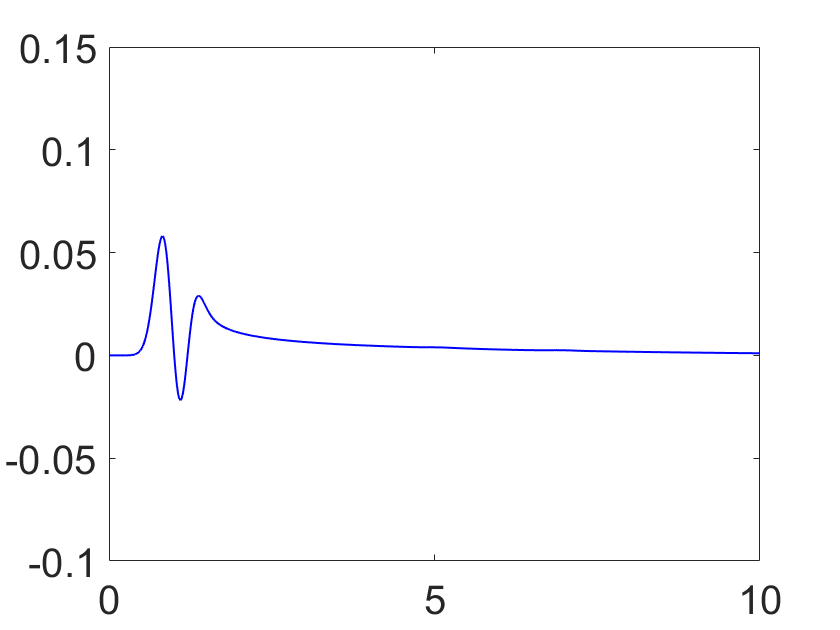}}
	\subfigure[signal with $\text{SNR}=15$ dB]{\includegraphics[width=0.45\textwidth]{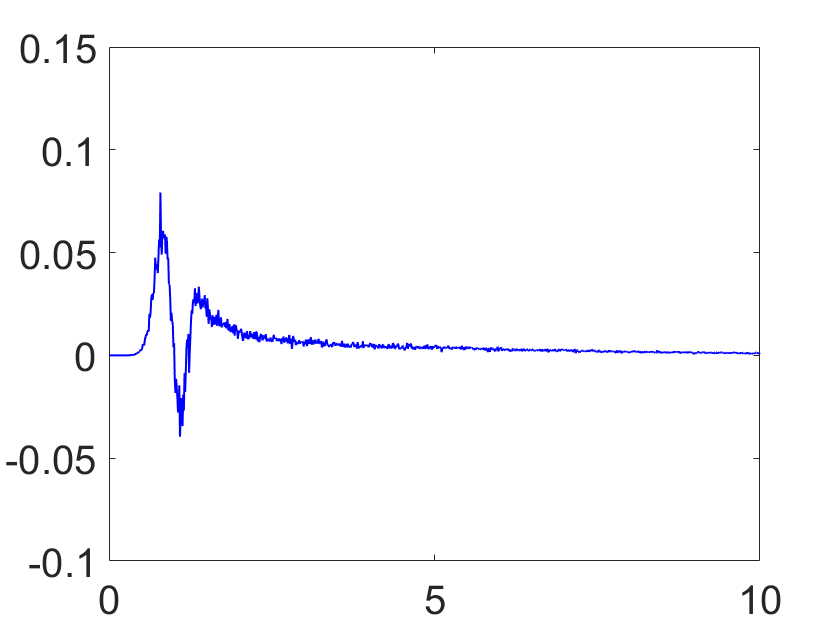}}\\
	\subfigure[signal with $\text{SNR}=5$ dB]{\includegraphics[width=0.45\textwidth]{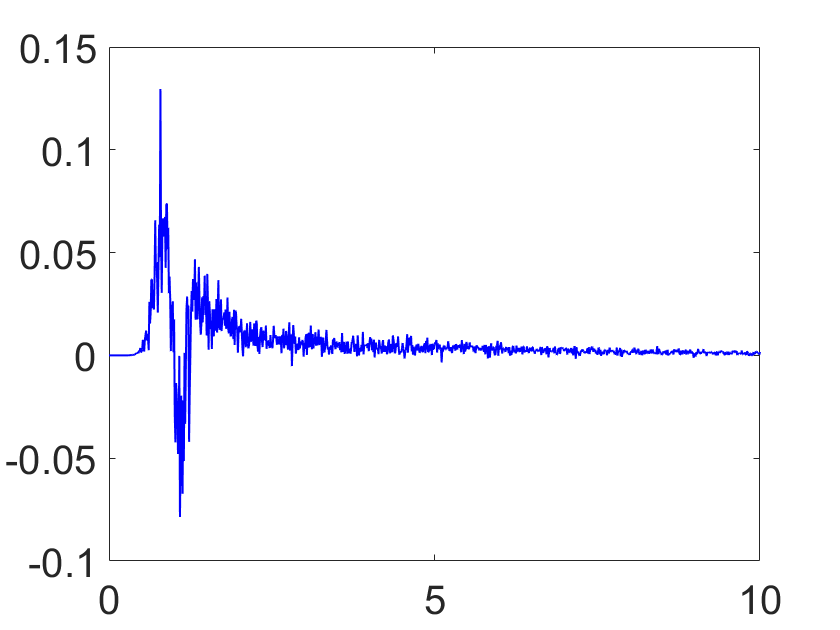}}		
	\subfigure[signal with $\text{SNR}=-1$ dB]{\includegraphics[width=0.45\textwidth]{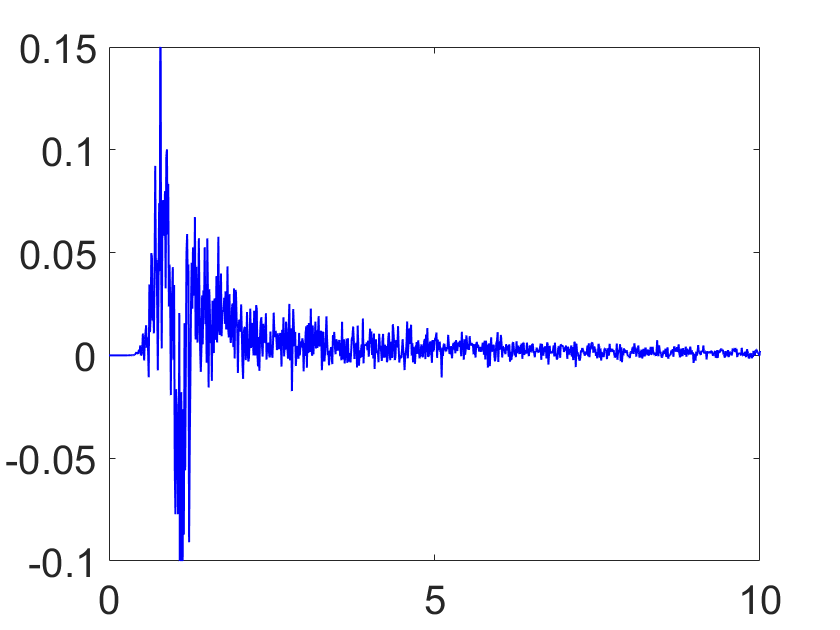}}
	\caption{\label{fig:signal} Time-dependent measurement signal at point $(1,0)$ under different noise levels.}
\end{figure}

\begin{figure}
	\centering
	\subfigure[exact source]{\includegraphics[width=0.45\textwidth]{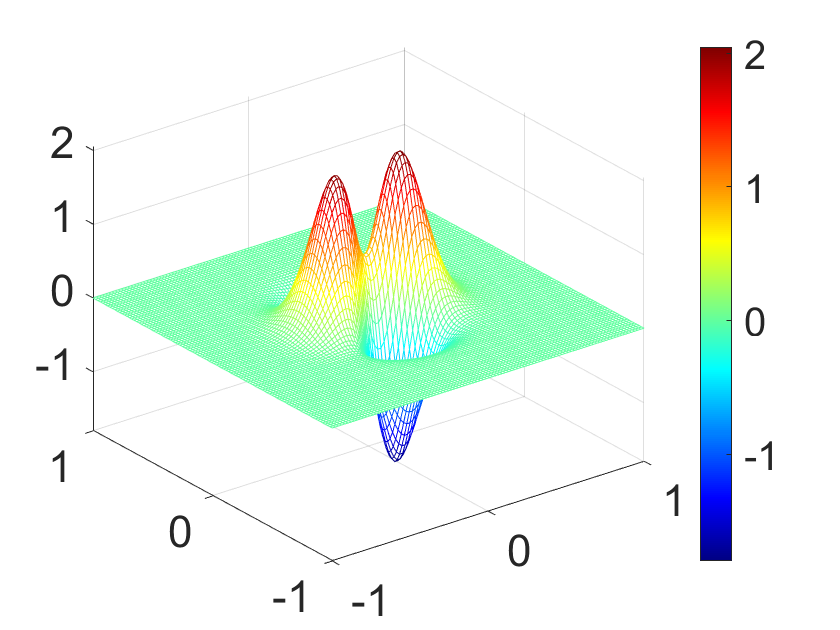}}
	\subfigure[reconstruction with $\text{SNR}=15$ dB]{\includegraphics[width=0.45\textwidth]{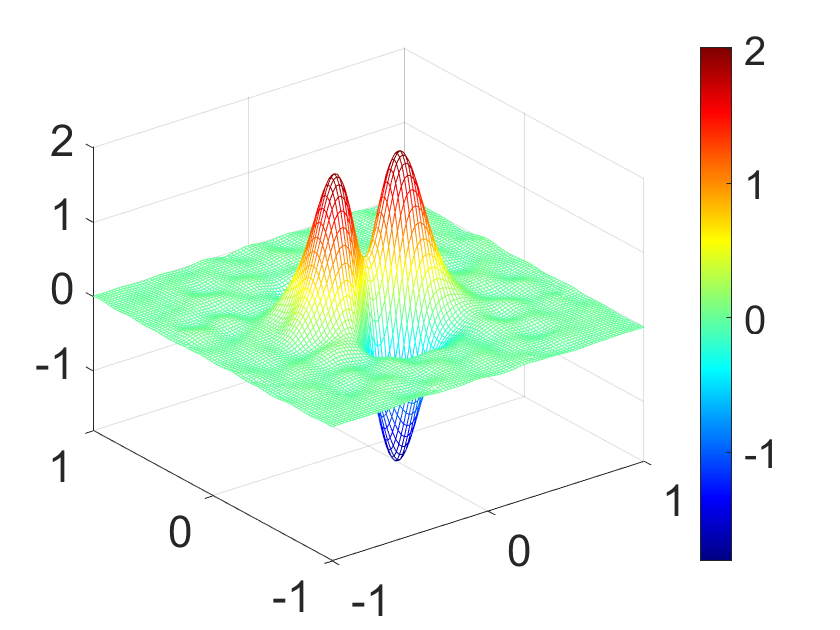}}\\
	\subfigure[reconstruction with $\text{SNR}=5$ dB]{\includegraphics[width=0.45\textwidth]{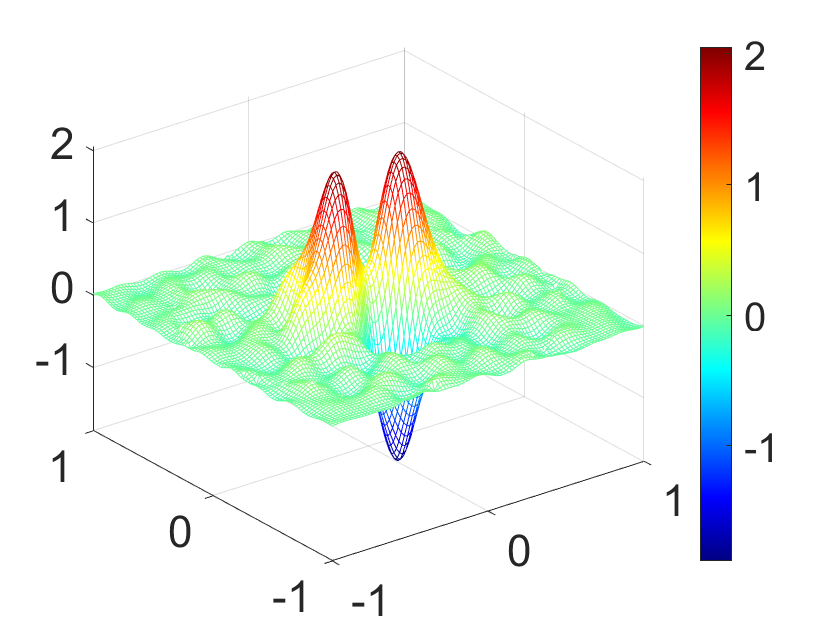}}
	\subfigure[reconstruction with $\text{SNR}=-1$ dB]{\includegraphics[width=0.45\textwidth]{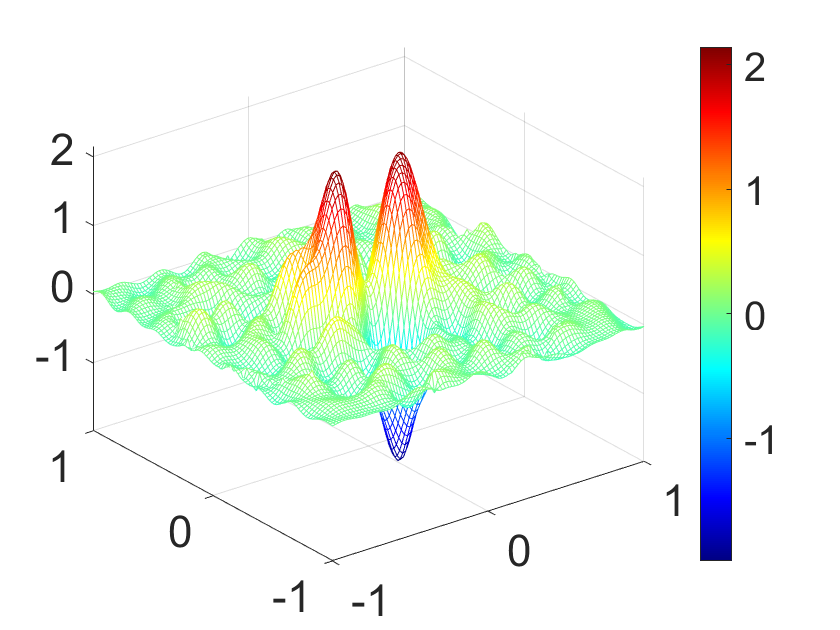}}
	\caption{\label{fig:near_2D}  Mesh plots of exact and reconstructed source functions  with frequency-domain imaging functional \eqref{I-near} under different noise levels in $2$D.}
\end{figure}

\subsection{\bf Example: Near2D}
	In the first example, we consider the following source function in $\R^2$,
	\begin{equation*}
		S(y_1,y_2)=1.1{\rm e}^{-30((y_1-0.01)^2+(y_2-0.12)^2)}-100(y_2^2-y_1^2){\rm e}^{-20(y_1^2+y_2^2)},
	\end{equation*}
    with compact support $\Omega=B_1(0)$,
	which is shown in Figure \ref{fig:near_2D}(a). 
	We place $80$ sensors uniformly on a circle of radius $r = 1$, and the time-dependent wave field $p$ is measured over the time interval $[0,\,10]$ with $1000$ steps. Figure \ref{fig:signal} shows the measured signal at the point $(1,0)$ under different noise levels. 
	
We then take the Fourier transform of time-dependent measurements $p(x,t)$ to obtain the frequency-domain data $u(x,k)$, where $k\in(0,30)$ with $100$ steps. 
Figure \ref{fig:near_2D} presents the exact and reconstructed source functions obtained using the indicator function $\mathcal{I}^{(2)}_{\mathrm{near}}(y)$ defined in \eqref{I-near}. We observe that the proposed method achieves accurate and stable reconstructions, even under extreme noise conditions where the signal is nearly fully corrupted (e.g., $\text{SNR} = -1$ dB).  Specifically, the relative error between the exact and reconstructed source functions is $6.53\%$ when $\text{SNR} = 15$ dB. 	
\begin{figure}
	\centering
	\subfigure[exact source]{\includegraphics[width=0.45\textwidth]{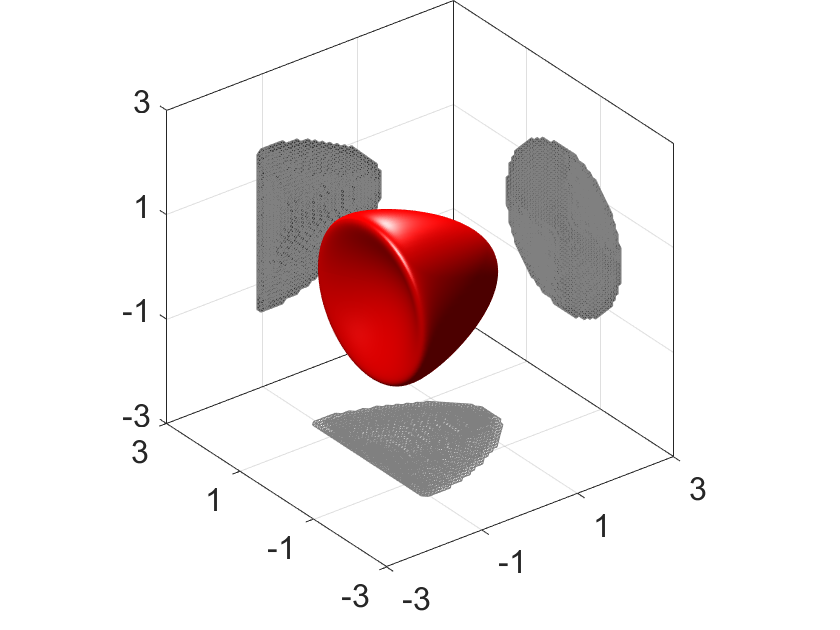}}
	\subfigure[reconstruction with $\text{SNR}=100$ dB]{\includegraphics[width=0.45\textwidth]{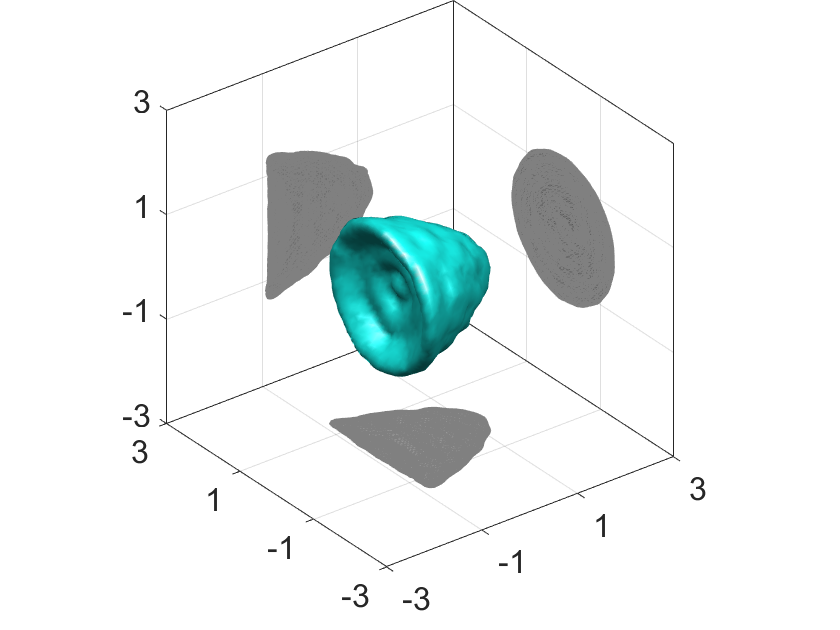}}\\
	\subfigure[reconstruction with $\text{SNR}=20$ dB]{\includegraphics[width=0.45\textwidth]{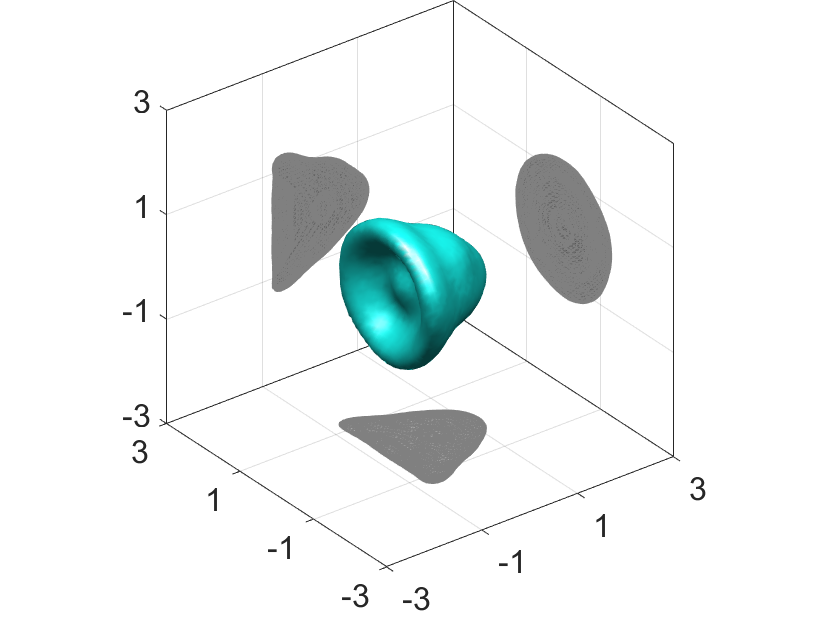}}
	\subfigure[reconstruction with $\text{SNR}=10$ dB]{\includegraphics[width=0.45\textwidth]{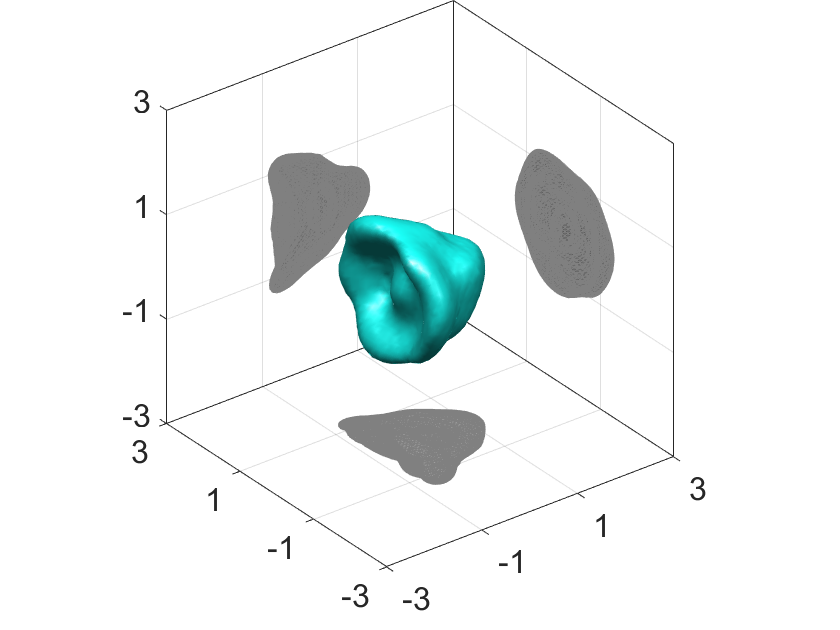}}
	\caption{\label{fig:near3D}  Iso-surface plots of the exact and reconstructed source functions in three dimensions, where the reconstruction is performed using the indicator function \(\mathcal{I}^{(1)}_{\mathrm{near}}(y)\) given in \eqref{eq:near3D_time} under different noise levels.}
\end{figure}

\subsection{\bf Example: Near3D}
	The second example is a source in $\R^3$ given by
    \begin{equation*}
		S(y)=
		\begin{cases}
			 1,  \quad   y\in  \Omega, \\
			 0,  \quad   y\notin  \Omega,\\
		\end{cases}
	\end{equation*}
    where the support $\Omega$ is parameterized as
	\begin{equation*}
		y(\theta, \phi)=(\cos \theta+0.65 \cos 2\theta-0.2, 1.5 \sin \theta \cos \phi, 1.5 \sin \theta \sin \phi), \ \ \theta\in(0, \pi), \phi\in(0, 2\pi).
	\end{equation*}
Figure \ref{fig:near3D}(a) shows the exact source. 
		We place $200$ sensors pseudo-equidistantly distributed on a sphere of radius $r = 3$ and measure time-dependent wave field $p$ over the time interval $[0,\,6]$ with $150$ steps. 
Figure \ref{fig:near3D} presents the iso-surface plots of exact and reconstructed source functions, obtained using the  imaging functional  $\mathcal{I}^{(1)}_{\mathrm{near}}(y) $ in \eqref{eq:near3D_time} under different noise levels, with the iso-surface value set to $0.9$. It can be observed that the proposed method exhibits excellent reconstruction performance at low noise levels; however, the reconstruction quality degrades as the noise level increases. In contrast to the stability observed in Example 1, the loss of robustness stems primarily from the discontinuity of the source function.

For comparison, we consider the alternative indicator $\mathcal{I}^{(2)}_{\mathrm{near}}(y)$ defined in \eqref{I-near}, 
where the wavenumber interval is taken as $k\in(0,30)$ with $100$ steps. From Figures \ref{fig:near3D} and \ref{fig:near3D_fre}, we observe that the time-domain imaging functional achieves better reconstruction performance than its frequency-domain counterpart under the same noise level. 
The  primary reason is that the time-domain formulation effectively integrates information over a continuous spectrum of frequencies, thereby enhancing stability and suppressing noise. In addition, the frequency-domain indicators  require multiplication by powers of $k$, which further amplify high-frequency noise.

\begin{figure}
	\centering
	\subfigure[reconstruction with $\text{SNR}=20$ dB]{\includegraphics[width=0.45\textwidth]{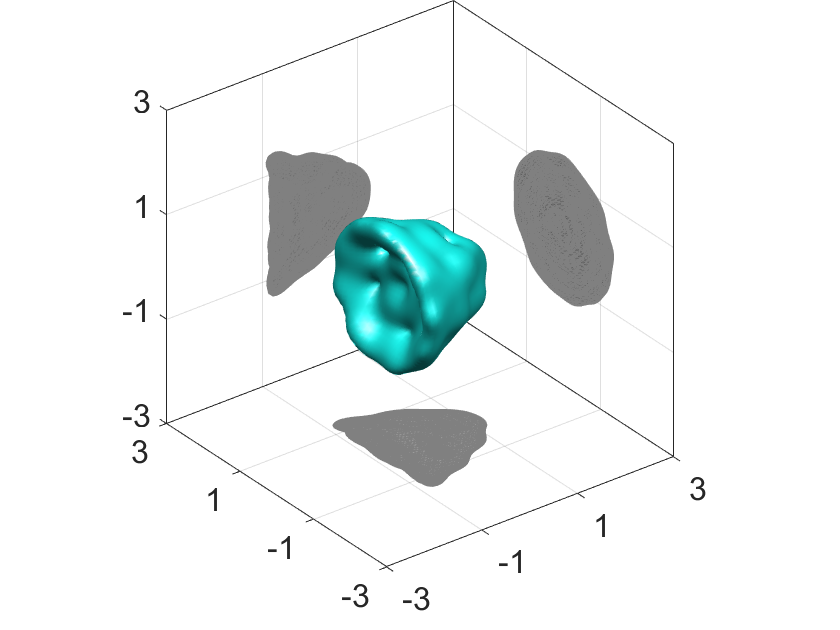}}
	\subfigure[reconstruction with $\text{SNR}=10$ dB]{\includegraphics[width=0.45\textwidth]{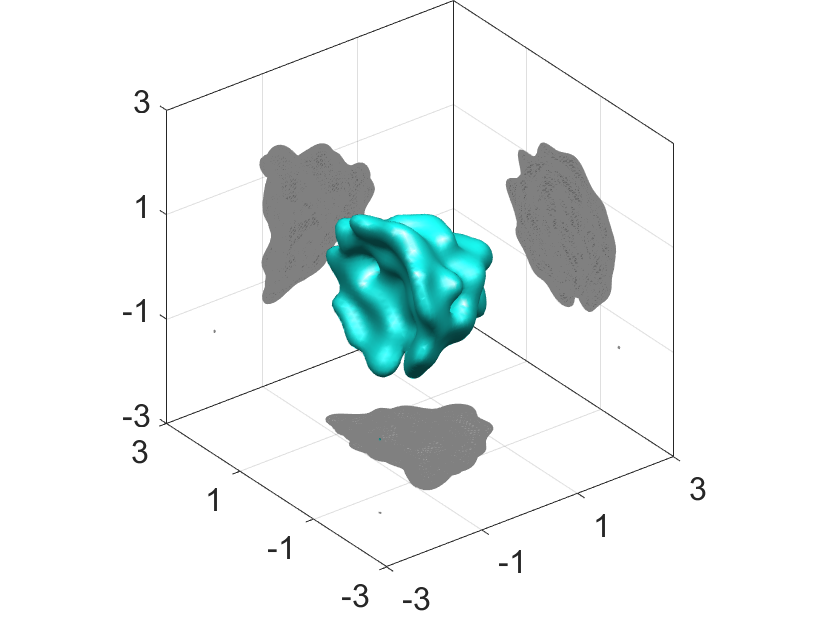}}
	\caption{\label{fig:near3D_fre} Iso-surface plots of the exact and reconstructed source functions in three dimensions, where the reconstruction is performed using the indicator function $\mathcal{I}^{(2)}_{\mathrm{near}}(y)$ given in \eqref{I-near} under different noise levels.}
\end{figure}

\begin{figure}
	\centering
	\subfigure[exact source]{\includegraphics[width=0.45\textwidth]{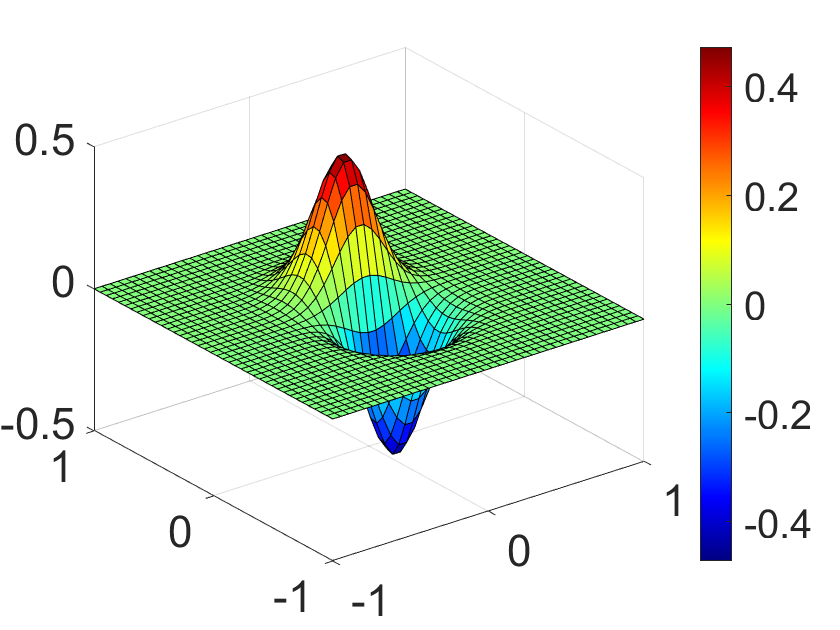}}
	\subfigure[$40\times 40$ directions over ${[-8,8]}^2$]{\includegraphics[width=0.45\textwidth]{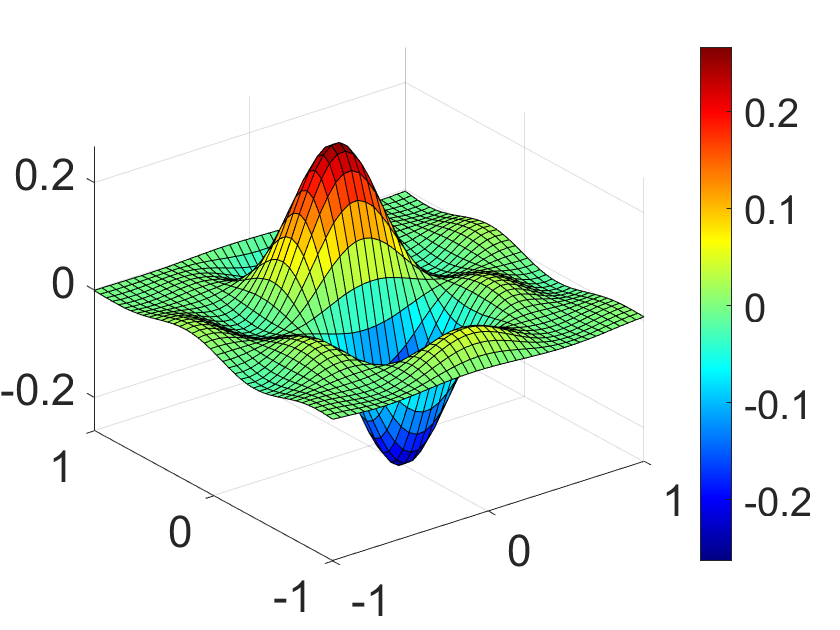}}\\
	\subfigure[$60\times 60$ directions over ${[-10,10]}^2$]{\includegraphics[width=0.45\textwidth]{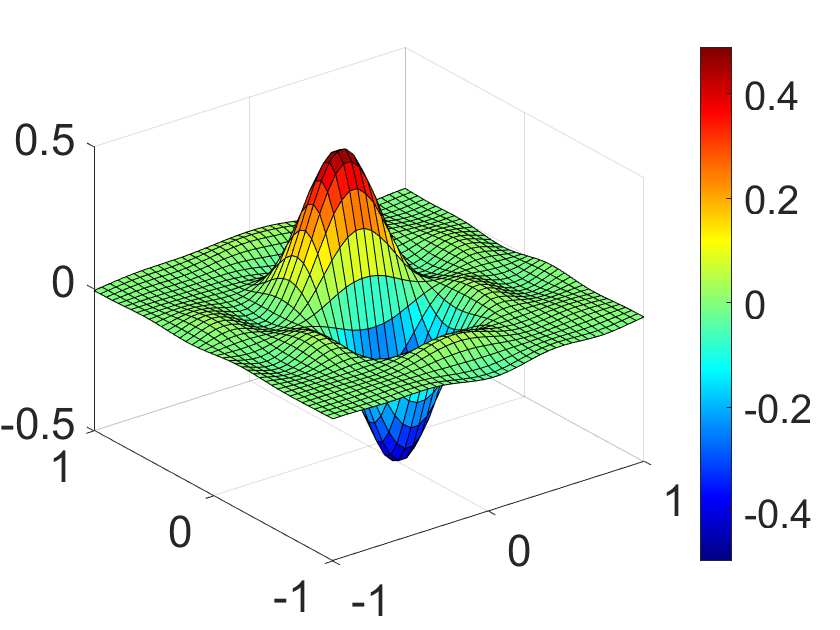}}
	\subfigure[$80\times 80$ directions over ${[-15,15]}^2$]{\includegraphics[width=0.45\textwidth]{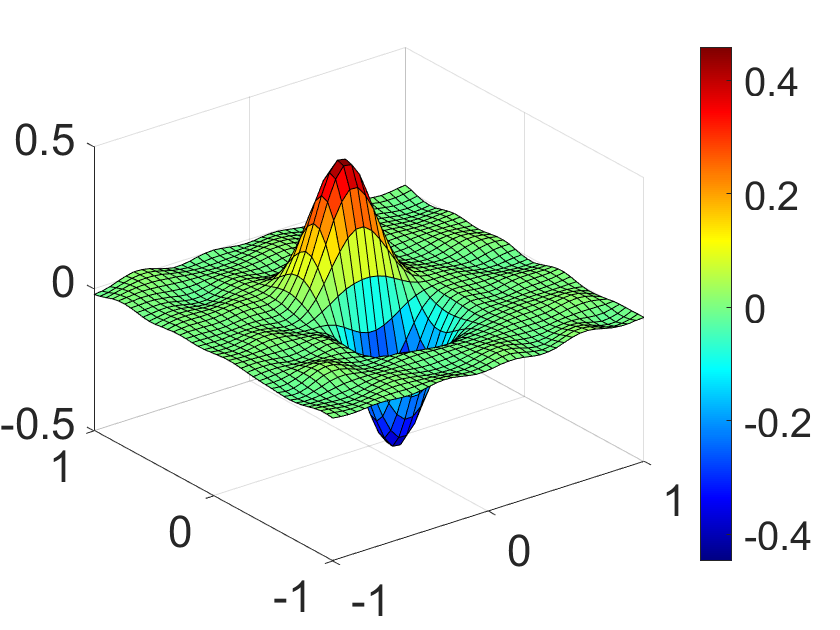}}
	\caption{\label{fig:far2D} Surface plots of exact and reconstructed source functions by plotting $\mathcal{I}_{\text{far}}$ given in \eqref{I-far} under different number of observation directions.}
\end{figure}

\subsection{\bf Example: Far2D}
The third source function is given by 
\begin{equation*}
	S(y_1,y_2)=\frac{1}{2}\left({\rm e}^{-20(y_1^2+(y_2-0.2)^2)}-{\rm e}^{-20(y_1^2+(y_2+0.2)^2)}\right)
\end{equation*}
with compact support $\Omega=B_1(0)$.
The  time-dependent far-field pattern $p^{\infty}$ is measured over the time interval $[-3,\,18]$ with $350$ steps. Moreover, noise corresponding to an $\text{SNR}=-1$ dB is also added to the measured data. 

The sampling mesh consists of a $40 \times 40$ uniform grid  covering the domain  $[-1,1]^2$. 
Figure \ref{fig:far2D} presents the reconstruction results using different sets of observation data. 
It is clear to see that the reconstruction quality improves with more observation directions.

\begin{figure}
	\centering
	\subfigure[exact source]{\includegraphics[width=0.45\textwidth]
		{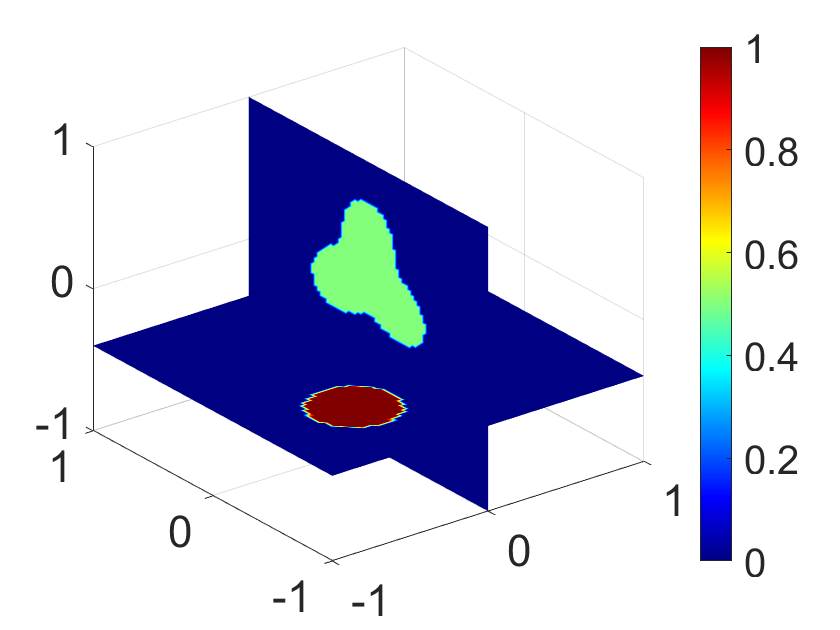}}
	\subfigure[reconstructed source]{\includegraphics[width=0.45\textwidth]
		{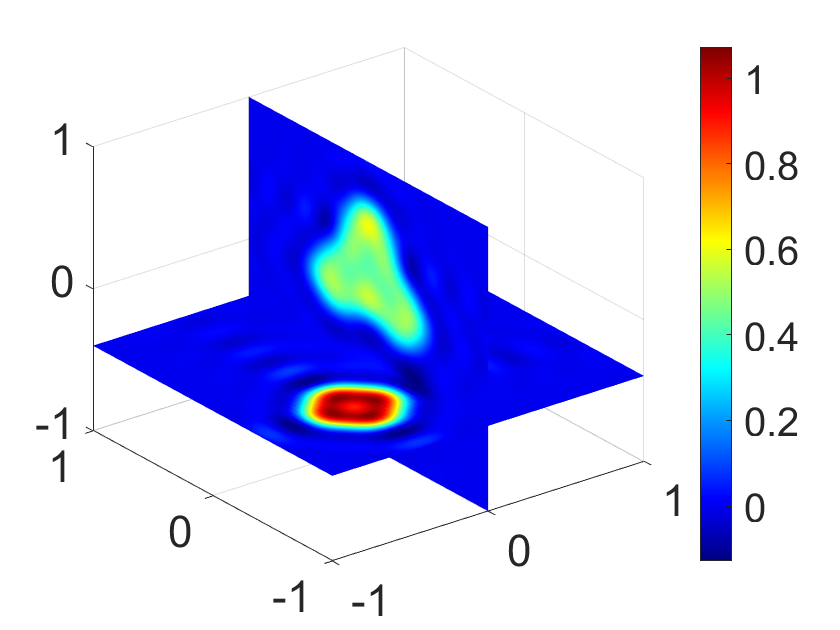}}
	\caption{\label{fig:far3D} Slice plots of the exact and reconstructed source functions obtained from $\mathcal{I}_{\text{far}}$ given in \eqref{I-far} with $\text{SNR}=-1$ dB. A $40 \times 40\times 40$ uniform grid is generated  covering the domain  $[-20,20]^3$.}
\end{figure}

\subsection{\bf Example: Far3D}
Finally, the fourth source function is defined as
\begin{equation*}
	S(y) =
	\begin{cases}
        1, & y \in \mathcal{B}, \\
		0.5, & y \in \mathcal{C}, \\
		0, & \text{otherwise},
	\end{cases}
\end{equation*}
where $\mathcal{B}$ is the ball centered at $(-0.4, -0.4, -0.4)$ with radius $0.4$ and 
$\mathcal{C}$ is a pear-shaped domain centered at $(0,0,0.2)$, given by the parametrization
\begin{equation*}
\mathcal{C}(\theta, \varphi) = \bigl(0.4 + 0.12 \cos 3\theta\bigr) \bigl(\sin\theta \cos\varphi,\ \sin\theta \sin\varphi,\ \cos\theta \bigr),
\end{equation*}
for $\theta \in [0, \pi]$ and $\varphi \in [0, 2\pi]$.
The  time-dependent far field $p^{\infty}$ is measured over the time interval $[-1.2,\, 1.2]$ with $60$ steps. 
The sampling mesh consists of a $60 \times 60\times 60$ uniform grid  covering the domain  $[-1,1]^3$. Figure \ref{fig:far3D} demonstrates that the proposed imaging functional achieves accurate reconstruction of discontinuous source functions in three dimensions using far-field measurements.

\section{Conclusion}
In this work, we have developed a novel quantitative direct sampling method for reconstructing the initial source distribution in acoustic wave equations from either near-field or far-field time-domain measurements. The proposed indicator functions are constructed via spacetime integrals of the measured wave field combined with appropriately designed auxiliary functions, which not only lead to a constructive uniqueness proof for the inverse initial value problem, but also yield a stable, efficient, and fully quantitative reconstruction scheme. In contrast to existing time-domain sampling methods, which are mostly qualitative, the present framework combines rigorous theoretical guarantees with practical computational efficiency, bridging the gap between frequency-domain analysis and real-time time-domain imaging. Extensive numerical examples further verify the accuracy, robustness against noise, and computational efficiency of the method, demonstrating its effectiveness in both near-field and far-field acoustic imaging scenarios.

Finally, we emphasize that the proposed mathematical framework has also been extended in the Appendix to treat unknown displacement sources. As a result, the method is capable of recovering either the initial velocity or the initial displacement distribution individually.
However, the current formulation does not support the simultaneous reconstruction of both initial velocity and displacement sources. Addressing this limitation represents a key research direction, which we will explore in our future work.


\appendix
\section{Inverse initial pressure source problems}
In this appendix, we extend  the scattering of initial velocity term \eqref{eq:main} to the initial displacement term, which is a common alternative configuration in acoustic imaging. We establish corresponding uniqueness results for both near-field and far-field measurements, complementing the main text.

Consider the time-domain wave equation for an initial displacement source $S(x)$
 (with zero initial velocity):
\begin{equation}\label{eq:main2}
	\begin{aligned}
		&\partial_{tt}p(x,t)-\Delta p(x,t)=0, \quad \quad \quad \quad (x,t)\in  \mathbb{R}^d \times \mathbb{R},\\
		& p(x, 0)=S(x), \quad \partial_t p(x,0)=0, \quad \ \ x\in \mathbb{R}^d.
	\end{aligned}
\end{equation}
Without loss of generality, we focus on the three-dimensional case and the two-dimensional case follows from analogous arguments.
By Duhamel's principle, the system \eqref{eq:main2} can be rewritten as
\begin{equation}\label{eq:Duham2}
	\begin{aligned}
		& \partial_{tt}p(x,t)-\Delta p(x,t)=\delta'(t) S(x),  \quad (x,t)\in  \mathbb{R}^3 \times \mathbb{R},\\
		& p(x, 0)=0, \quad \partial_t p(x,0)=0,  \quad \quad \ \ x\in \mathbb{R}^3
	\end{aligned}
\end{equation}
where $\delta'(t)$ denotes the distributional derivative of the Dirac delta function.
Thus, the wave field can be represented by
 \begin{equation*}
		p(x,t)
		=\int_{-\infty}^{\infty}\int_{\mathbb{R}^3} 
		\frac{\delta(t-\tau-|x-y|)}{4\pi|x-y|}  \delta'(\tau)   S(y)\, \mathrm{d}y\, \mathrm{d}\tau =\int_{\mathbb{R}^3} \frac{\delta'(t-|x-y|)}{4\pi|x-y|}  S(y)\, \mathrm{d}y.
\end{equation*}
Correspondingly, by  similar arguments to Theorem \ref{thm:far-field}, one can get the far-field pattern
\begin{equation}\label{eq:Fourier2}
	p^{\infty}(\hat x, t)=\frac{1}{4\pi}\int_{\mathbb{R}^3} \delta'(t+\hat x \cdot y) S(y)\, \mathrm{d}y,
\end{equation}

Applying the Fourier transform to the  time-dependent wave field $p(x,\cdot)$, 
then the time domain system \eqref{eq:Duham2} reduces to the Helmholtz equation with the Sommerfield radiation condition
\begin{equation*}
	\begin{aligned}
		& -\Delta u- k^2 u=\mathrm{i} k S  \quad\mbox{in} \,\mathbb{R}^3,\\
		& \lim_{r\to\infty} r
		\left( \frac{\partial u}{\partial r}-\mathrm{i}ku \right)=0, \qquad \quad \qquad \ r=|x|.
	\end{aligned}
\end{equation*}
The solution to the last equation is given by 
	\begin{equation}\label{eq:fre-solution-A}
	u(x,k)=\mathrm{i} k  \int_{\mathbb{R}^3}  \Phi(x,z; k) \,  S(z) \, \mathrm{d}z,\quad (x,k)\in  \mathbb{R}^3 \times \mathbb{R}^+.
\end{equation}

Now we present the uniqueness results for reconstructing  the initial displacement source $S$ from the time-domain near-field and far-field measurements, respectively.
\begin{theorem}
The initial displacement $S\in C_c^{\infty}(\Omega)$ can be represented by 
\begin{equation}\label{eq:Source-A}
		S(y) = 2 \int_{\Gamma} \langle
        \partial_{\nu_x} G(x,y; \cdot), p(x,\cdot)\rangle
		\,\mathrm{d}s(x),\quad y\in\R^3.
	\end{equation}
\end{theorem}
\begin{proof}
Since $\partial_{\nu_x}  G(x, y;  \cdot)\in \mathscr{S}(\mathbb{R})$  and $p(x,\cdot) \in C_c^{\infty}(\mathbb{R}) \subset \mathscr{S}(\mathbb{R})$, 
using the  Fourier transform of tempered distributions  and the fundamental solution \eqref{eq:Fourier-fund},  the right-hand side of  \eqref{eq:Source-A} is given by
\begin{equation*}
	\begin{aligned}
		 \text{RHS}&=2 \int_{\Gamma} \frac{1}{2\pi}\int_{-\infty}^{\infty}   \partial_{\nu_x}\left(\frac{\mathrm{i}k}{4\pi}h_0^{(1)}(k|x-y|) \right)\, \overline{u} (x,k)\mathrm{d}k \, \mathrm{d}s(x)\\
		&=\frac{1}{2\pi^2}   \int_{0}^{\infty}  \int_{\Gamma}   \Re\left[ \mathrm{i} k \, \partial_{\nu_x} h_0^{(1)}(k|x-y|) \, \overline{u}(x,k) \right]   \mathrm{d}s(x) \,\mathrm{d}k, 
	\end{aligned}
\end{equation*}
where the last justified by Fubini's theorem as $u(x,\cdot) \in \mathscr{S}(\mathbb{R})$. 

From the representation \eqref{eq:fre-solution-A}, one gets
	\begin{equation*}
 \Re (u(x,k)) =\frac{-k^2}{4\pi}\int_{\mathbb{R}^3} j_0(k|x-z|) S(z) \, \mathrm{d}z,
\end{equation*}
together with \eqref{eq:fre-solution-A}, we obtain
\begin{equation*}
\begin{aligned}
& \text{RHS}=\frac{1}{2\pi^2}   \int_{0}^{\infty}  \int_{\Gamma}   \left[ -\mathrm{i}  \, \partial_{\nu_x} j_0(k|x-y|) u(x,k) + \mathrm{i} \Re(u(x,k)) \partial_{\nu_x}h_0^{(1)}(k|x-y|) \right] k \, \mathrm{d}s(x) \,\mathrm{d}k \\
&= \frac{1}{2\pi^2} \int_{0}^{\infty} \int_{\Gamma} \int_{\mathbb{R}^3} \left[\frac{\partial j_0(k|x-y|)}{\partial \nu_x}\Phi(x,z;k) - j_0(k|x-z|) \frac{\partial \Phi(x, y; k) }{\partial \nu_x}\right]  S(z)\, \mathrm{d}z\, \mathrm{d}s(x)\, k^2\, \mathrm{d}k\\
&= S(y),\quad y\in\R^3,
\end{aligned}
\end{equation*}
where the last equality follows directly from the argument used in the proof of Theorem~\ref{thm:near3D}.  
This completes the proof. 
\end{proof}

\begin{theorem}
Given the time-domain far-field pattern $p^\infty$, the initial displacement  $S\in C_c(\Omega)$ in $\R^3$ admits the representation
	\begin{equation}\label{eq:indicator-far3D-A}
		S(y) =\frac{1}{2\pi^2}  \int_{\mathbb{R}^3} \int_{-\infty}^{\infty}  p^{\infty}\left(\hat x, t-\hat x\cdot y \right) \frac{\mathrm{i} \mathrm{e}^{\mathrm{i}|x|t }}{|x|}\, \mathrm{d}t\, \mathrm{d}x,\quad y\in\R^3.
	\end{equation}
\end{theorem}
\begin{proof}
Substituting the far-field pattern \eqref{eq:Fourier2} into the right-hand side of \eqref{eq:indicator-far3D-A}, it yields  
	\begin{equation*}
		\begin{aligned}
			\text{RHS}=\frac{1}{2\pi^2}  \int_{\mathbb{R}^3} \int_{-\infty}^{\infty}  \left(\frac{1}{4\pi}\int_{\mathbb{R}^3} \delta'(t+ \hat x \cdot z - \hat x \cdot y) S(z)\, \mathrm{d}z\right) \frac{\mathrm{i} \mathrm{e}^{\mathrm{i}|x|t }}{|x|} \, \mathrm{d}t\, \mathrm{d}x. 
		\end{aligned}
	\end{equation*}
Using the distributional identity
\begin{equation}\label{eq:distri-identity-A}
\int_{-\infty}^{\infty}\delta'(t+a){\rm e}^{{\rm i}|x|t}\, {\rm d}t
=-{\rm i}|x|{\rm e}^{-{\rm i}|x|a},
\end{equation}
the factor $-{\rm i}|x|$ cancels ${\rm i}/|x|$, yielding
\begin{equation*}
\mathrm{RHS}
=\frac{1}{(2\pi)^3}
\int_{\mathbb{R}^3}\int_{\mathbb{R}^3}
{\rm e}^{-{\rm i}x\cdot(z-y)}S(z)\,{\rm d}z\,{\rm d}x=S(y),\quad y\in\R^3.
\end{equation*}
\end{proof}
\begin{remark}
	Motivated by the above constructive theorem proofs, one may reconstruct the initial displacement by using the following indicators
	\begin{equation*}
		\mathcal{J}_{\text{near}}(y) = 2 \int_{\Gamma} \langle
        \partial_{\nu_x} G(x,y; \cdot), p(x,\cdot)\rangle
		\,\mathrm{d}s(x),\quad  y\in D, 
	\end{equation*}
	and 
	\begin{equation*}
		\mathcal{J}_{\text{far}}(y) =\frac{1}{2\pi^2}  \int_{\mathbb{R}^3} \int_{-\infty}^{\infty}  p^{\infty}\left(\hat x, t-\hat x\cdot y \right) \frac{\mathrm{i} \mathrm{e}^{\mathrm{i}|x|t }}{|x|}\, \mathrm{d}t\, \mathrm{d}x, \quad  y\in D.
	\end{equation*}
Moreover, we emphasize that the singularity of the imaging functional $\mathcal{J}_{\mathrm{far}}$ at $|x|=0$ is removable, as the factor $|x|$ arising from the distributional identity \eqref{eq:distri-identity-A} exactly cancels the prefactor $1/|x|$.
\end{remark}

\section*{Acknowledgement}
The research of X. Liu is supported by the National Key R\&D Program of China under grant 2024YFA1012303 and the NNSF of China under grant 12371430.
The work of  X. Wang is supported by NSFC grant 12471397 and Heilongjiang
Provincial Natural Science Foundation grant YQ2024A003.

\end{document}